\documentclass[lettersize,journal]{IEEEtran}
\usepackage{booktabs} 
\usepackage[ruled]{algorithm2e} % For algorithms
\usepackage{amsthm}

\usepackage[english]{babel}
\usepackage{blindtext}
\usepackage{color}
\usepackage{xspace}
\usepackage{caption}
\usepackage{multirow}
\usepackage{amsmath, amssymb}  
\usepackage{graphicx}
\usepackage{gensymb}
\usepackage{makecell}
\usepackage{enumerate}
\usepackage{circledsteps}
\usepackage{pifont}
\usepackage{arydshln}

\usepackage{subfigure}
\usepackage[compact]{titlesec}
\usepackage{balance}
\usepackage{tikz}
\usepackage{listings}
\usepackage{makecell, tabularx}
\usepackage{autonum}
\usepackage{cleveref}
\usepackage{amsmath}

\SetAlFnt{\small}
\SetAlCapFnt{\small}
\SetAlCapNameFnt{\small}
\SetAlCapHSkip{0pt}
\IncMargin{-\parindent}

\newcommand{\ie}{\emph{i.e.,}\xspace}
\newcommand{\eg}{\emph{e.g.,}\xspace}
\newcommand{\etal}{\emph{et al.}\xspace}

\newtheorem{theorem}{Theorem}

\begin{document}
\title{NFT Games: an Empirical Look into the Play-to-Earn Model}

\author{Yixiao Gao,
        Fei Li,
        Ruizhe Shi,
        Ruizhi Cheng,
        Jean Zhang,
        Bo Han,
        Songqing Chen
\thanks{Y. Gao. F. Li, R. Shi, R. Cheng, B. Han, and S. Chen are with the Department of Computer Science, 
George Mason University, Fairfax, VA 22030, the United States.}
\thanks{J. Zhang is with the School of Business, 
Virginia Commonwealth University, Richmond, VA 23284, the United States.}
}

\maketitle

\begin{abstract}
The past decade has witnessed the burgeoning and continuous development of blockchain and its applications. Besides various cryptocurrencies, an industry that has quickly embraced this trend is gaming. Thanks to the support of blockchain, games have started to incorporate non-fungible tokens (NFTs) that can enable a {i.e., new gaming model}, play-to-earn (P2E), which incentivizes users to participate and play. While recent studies looked at several NFT games qualitatively and individually, an in-depth understanding is still missing, particularly on how the P2E model has transformed traditional games. In this work, we set to conduct a measurement study of NFT games, aiming to gain a comprehensive understanding of the effectiveness of P2E in practice. For this purpose, we collect and analyze relevant NFT transaction data from the underlying blockchain (e.g., Ethereum) of 12 games, supplemented with various data scraped from their websites. Our study shows that (1) a few top wallets control unproportionally high percentage of NFTs, and the majority of wallets own only one or two NFTs and do not actively trade; (2) promotion events do boost the trade amount and the NFT price for some games, but their effect does not sustain; and (3) few players actually earned a profit, and players in 9 out of 12 games who traded NFTs have a negative profit on average. 
Motivated by these findings, we further investigate effective incentive mechanisms based on game theory to improve the trading profits that players can earn from these NFT games. Both modeling and simulation results confirm the effectiveness of the proposed incentive mechanism. 
\end{abstract}


\section{Introduction}

Over the last decade, the token-based economy has significantly evolved in the form of cryptocurrencies and decentralized applications (DApps) \cite{bitcoin2008,ethereumwhitepaper}. The inception of {\em Metaverse}-related concepts further created a need for digital assets that are compatible with the token-based economy, thus popularizing the market of non-fungible tokens (NFTs)~\cite{borri2022economics}. A report states that the NFT Market size was valued at USD 26.9 Billion in 2024~\cite{report}. NFTs encapsulate the digital representation of assets such as artwork or game props\footnote{Game props refer to a 3D representation of an object, typically used in online gaming.}, and they are transferred over the publicly verifiable and immutable blockchain ledger. 

NFTs have garnered notable attention since  
the spike of CryptoKitties \cite{cryptokitties} in 2017. As a result, various NFT marketplaces have emerged that enable NFT trading on their platforms. These marketplaces also include gaming markets that use NFT tokens as gaming assets. Since July 2020 \cite{6boonparn2022social}, the gaming industry has generated the highest NFT trading volume, accounting for $\sim$51\% of the industry usage of blockchain as of August 2022 \cite{430-3}.

The prosperity of NFT games can be largely attributed to the underlying blockchain, which enables a new business model called play-to-earn (P2E).  P2E takes advantage of virtual currency and the potential appreciation of NFT tokens. In P2E, players can convert profits into virtual currency on the blockchain by selling NFTs or by using the attributes of NFTs in the game to achieve value growth. For example, in CryptoKitties, Kitties (NFTs) can be used to breed more Kitties that can be sold as cryptocurrencies, %and then 
which are then converted into currency in the real world. Therefore, the support of blockchain enables a new business model (P2E) that allows users to {\em earn} money on gaming platforms in a more transparent and user-owned way. P2E is distinctly different from the legacy gaming models such as pay-to-play (P2P) where users typically {\em spend} money on the gaming platform. As a result, P2E has fundamentally transformed the business model of games, providing better incentives to players and great potential for the game industry.

All these factors highlight the growing importance of NFT games in the token-based economy. As such, it warrants a formal analysis to understand the key characteristics of NFT games and their business model. We note that there are limited efforts to study the anatomy of NFT games at scale, especially within the context of their incentive mechanisms. Notable prior works~\cite{5vidal2022new,7francisco2022perception,alam2022understanding} mainly focus on the social media footprint of NFT games or the characteristics of one particular game ({\em\eg} Axie Infinity). Therefore, the research gap lies in studying the NFT game ecosystem at scale with an emphasis on its evolving business model. Closing this gap requires a comprehensive measurement study across the aforementioned frontiers to fully understand the evolving trends in NFT games and the token-based economy at large. 

In this paper, we make the attempt to close this research gap with various NFT games. For this purpose, we conduct a measurement study across 12 NFT games (all on Ethereum but one from Ronin) since their first debut to quantitatively understand the P2E model. Our analysis and findings are summarized below as key contributions. 

\begin{itemize}

\item \textbf{Ownership Dominance of NFTs.} Few top wallets own a significant portion of NFTs, while the other wallets often own one or several NFTs, and they do not actively conduct NFT trading. 

\item \textbf{Varying Effect of Promotion Events.} The promotion events do boost more transactions (e.g., a day's transaction is up to 27\% of total trades) and higher NFT prices for most games, but the effect of promotion events does not sustain. 

\item \textbf{Few Players Earned.} Few players actually earned profit, where in 9 (out of 12) games the players who traded NFTs have a negative profit on average.

\end{itemize}

By fitting the empirical data, we propose a simple  model to characterize the profitability of these NFT games. Users can utilize this model to assess the  risk before participating in a particular NFT game in the future.

As our findings show that the current P2E model does not work well for users, we are inspired to investigate if proper incentive mechanisms could help improving the trading profits of the NFTs in these games. Accordingly, we abstract the player activities in the NFT games and propose an incentive mechanism model based on game theory. Bot the modeling and simulation results show that our proposed appropriate incentive mechanism can significantly improve the profits that players earn from the NFT games. 

\textbf{Ethical Consideration.} All the data collected and analyzed in this study are public. The wallets are pseudonymous, and no individual wallet addresses are identified or presented in this paper. Our analysis is conducted at an aggregated level and no potential deanonymization. Therefore, there is no ethical issue. 

The rest of the paper is organized as follows. We present the background information in Section~\ref{sec:background}. We present an overview and our detailed analysis of P2E in Section~\ref{sec:methodology} and Section~\ref{sec:p2e}, respectively. We abstract the current NFT P2E in Section \ref{sec:extend}. We investigate the incentive mechanism in Section \ref{sec:tla}, and make concluding remarks in Section~\ref{sec:conclusion}.

\section{Background}
\label{sec:background}

In this section, we provide a brief background about the NFT gaming platforms. In particular, we discuss the token types used in NFT gaming platforms. 

NFTs (Non-Fungible Tokens)~\cite{nadini2021mapping} are considered digital assets whose ownership certificates and/or data are stored on the blockchain. Physical assets in the real world (\eg vehicles, artworks, and shoes) or digital assets in the real or virtual world (\eg music, video, images, and game props) can be tokenized as NFTs. As player-to-player trading is a significant part of the game community, the advantages of blockchains are noticed by the gaming industry. By tokenizing the game props and in-game assets to NFTs, all the related ownership certificates on the blockchains become immutable; and all the transaction history is retained distributive.

Before introducing NFTs, Ethereum first released a token standard called "fungible tokens" (ERC-20)~\cite{erc20}, which is similar to Bitcoin~\cite{bitcoin2008}, with additional smart contract compatible functionalities. In 2018, Ethereum introduced the Non-Fungible Token Standard (ERC-721)~\cite{erc721} that popularized the concept of NFTs and was also adopted by the gaming community. However, game developers soon found that ERC-721 does not cater to all gaming requirements. For instance, in a gaming area, there are multiple game props (\eg skin, and weapons) that are non-fungible as an item but fungible among the copies of the same item. Moreover, ERC-721~\cite{erc721} mandates each NFT on a separate transaction when being transferred, which adds costs and reduces efficiency. To overcome such limitations, Ethereum then introduced ERC-1155~\cite{erc1155} that combines the ERC-20~\cite{erc20} and ERC-721~\cite{erc721} standards and allows multiple tokens storage within one contract. 

At present, NFT games use both ERC-721 and ERC--1155 standards in their games. This is to be noted that there is no standardized term that specifies the characteristics of an NFT game. In fact, several terms are used in the community to refer to NFT games, including NFT games~\cite{7francisco2022perception}~\cite{6boonparn2022social}, NFT-based games~\cite{430-1}, play-to-earn games~\cite{4delfabbro2022understanding}, and blockchain games~\cite{alam2022understanding}. For the ease of analysis, we standardize the term NFT games based on the key characteristics common among the all terms that are broadly used for a gaming platform that supports NFT tokens. In the following, we enumerate standardized specifications of a gaming platform which we then use as blueprint to define an NFT game. 

\begin{enumerate}
        \item The gaming platform supports at least one type of NFT token including ERC-721 or ERC-1155. 
        \item NFT tokens used on the platform are tradable, and the transaction data is recorded on the blockchain.
        \item NFT interaction is a mandatory part of the gaming experience. Players cannot bypass the NFT token interactions during the game. 
        \item The gaming platform provisions at least one form of token incentive to the players which may be in the form of NFTs or other cryptocurrency tokens. 
\end{enumerate}

We use the aforementioned specifications to define an NFT game, and follow them during our data collection process and analysis. Note that specification (4) states that gaming platforms support incentive payouts in the form of NFT or legacy tokens. By taking a closer look at this specification across real-world NFT games, we found that NFT games have two main token types as part of the incentive process. 

\begin{itemize}
    \item \textbf{Property Tokens}: Property tokens refer to digital assets on the gaming platform that follow ERC-721~\cite{erc721} and ERC-1155~\cite{erc1155} standards on Ethereum or other blockchains. As mentioned above, NFTs usually represent unique digital assets in the game; assets vary based on the game context, such as in-game props, characters, cards, virtual land, and artworks.
    
    \item \textbf{Utility Tokens}: Utility tokens represent the fungible tokens (\ie ERC-20 on Ethereum or identical tokens on other chains) in the game. Utility tokens usually perform functions such as currency, materials, and votes. Specifically, utility tokens with voting power are also known as \emph{governance tokens}. Early NFT games directly used Ether as a utility token. As the gaming content gradually evolved, Ether could not satisfy all requirements expected from a utility token. Therefore, gaming companies standard developing their own utility token, typically following the ERC-20 standard. 
\end{itemize}

NFT games often employ a play-to-earn (P2E) economic model to incentivize the players to participate.  Compared to the pay-to-play and free-to-play in traditional games, players in NFT games can freely trade NFTs and thus have the potential to make profits while playing. This is different from traditional games where the trading of gaming assets is largely restricted by game companies, while there is no guarantee on private transactions between players. Thus, play-to-earn is expected to fundamentally transform traditional games and boost the game industry. 
\section{Overview of our Measurement Study}
\label{sec:methodology}

In this section, we present our methodology and a summary of the real-world data that we have collected to investigate the effectiveness of the play-to-earn model.

\subsection{Methodology}
 
Since all the NFT transaction data is on the blockchains, the data is publicly available for us to collect. The transaction data until 3/31/2023 is retrieved directly from our Ethereum node, while we expand the period of the data set to 6/30/2023 and download the additional data from the block explorer, Etherscan~\cite{etherscan}. The data of Axie Infinity is on the Ronin~\cite{ronin} sidechain, and we collected ERC-721~\cite{erc721} trade history through the REST API~\cite{roninAPI}. To convert the trade price to USD, we retrieved the daily exchange rate between the ERC-20 tokens to USD and ETH/WETH to USD from Yahoo! Finance~\cite{yahoo}. 

We collected the data from Day 1 of each game. Thus 
the start dates vary for different games since each game was launched at a different time. The end date of all the blockchain data we retrieved is 06/30/2023 except for Axie Infinity, for which we only have its data till 03/31/2023 
since they stopped REST API support for data access.  

After we convert the NFT trading value into USD, we track the trading history of every NFT from the first day it was minted. Since the development cost and the supply number vary from different NFT collections, calculating the original cost of NFTs when minted on the blockchain is complex and hard to quantify. To simplify the problem, 
we set the initial cost when the game developer mints the NFT as 0. Thus, for each NFT trade, we set the trading price of NFT in the previous transaction as the cost to calculate the profit or loss. NFTs, as the game props, are sold from the game developer teams and are then circulated between players. Thus, we consider the first trade of an NFT as the sale from the game developer team and the second and all the subsequent trades of the NFT as the trades between the players.
Based on the trading history of each NFT, we can explore the profit and loss from each NFT's perspective.

\subsection{Data Overview}

\begin{table}
\footnotesize
\centering
\caption{Summary of Ethereum-based NFT Games \normalfont{(Axie Infinity is on Ronin, a side chain of Ethereum)}. \normalfont{In the table, NFT Holders refer to the total number of wallets. Since NFTs of both Decentraland and Sandbox are named LAND, we use {\bf LANDD} and {\bf LANDS} to differentiate them in our study.}}
\label{tab:games_eth}
\begin{tabular}{c|c|c|c}
Game (NFT) & Transactions & \makecell{NFT \\ Total Supply} & NFT Holders \\ \hline
Benji Bananas (BENJI) & 12,729 &5,000 & 2,937 \\\hline
Blankos Block Party (BLNKS) & 4,863 & 4,870 & 1,417 \\\hline
CryptoKitties (CK) & 5,620,779 &2,022,717 & 120,334 \\\hline
Decentraland (LANDD) & 216,073 & 92,598 & 5,699 \\\hline
Ethermon (EMONA) & 58,330 & 56,536 & 5,640 \\\hline
F1 Delta Time (F1DTI) & 84,685 & 40,572 & 1,438 \\\hline
League of Kingdoms (LOKL) & 16,107 & 7,817 & 2,030 \\\hline
League of Kingdoms (LOKR) & 5,295 & 1,027 & 737 \\\hline
My Crypto Heroes (MCHH) & 147,448 & 20,997 & 1,010 \\\hline
Sorare (SOR) & 1,193,780 & 336,353 & 16,745 \\\hline
Spider Tanks (TANKS) & 8,919 & 653 & 368 \\\hline
Sandbox (LANDS) & 224,138 & 76,080 & 15,988 \\\hline
Axie Infinity
(AXIE) & 438,418,662 & 7,334,880 & 1,949,476\\
\end{tabular}
\vspace{-20pt}
\end{table}

With all the data collected, 
Table~\ref{tab:games_eth} summarizes the NFT token contract transaction statistics of the NFT games. Note that all activities on the blockchain are executed via transactions. Transactions include not only NFT trades but also NFT minting, auction, bidding, etc. We collected the NFT trades from the token smart contracts of these NFT games.

In Table~\ref{tab:games_eth}, \textbf{Decentraland}~\cite{decentraland}, \textbf{League of Kingdoms}~\cite{leagueofkingdoms}, and \textbf{Sandbox}~\cite{sandbox} are virtual land and metaverse. Virtual lands and their accessories such as virtual buildings and virtual decoration products are NFTs in these games, however, virtual lands would be the most significant one since other NFTs are dependent on the lands. Note \textbf{League of Kingdoms} has two NFT collections in the game but their functionalities in the game are independent. Thus, we treat each as a game and analyze them separately.

\textbf{Axie Infinity}~\cite{Axie}, \textbf{Benji Bananas}~\cite{benjibananas}, \textbf{Blankos Block Party}~\cite{blankos}, \textbf{Ethermon}~\cite{ethermon}, \textbf{F1 Delta Time}~\cite{F1DeltaTime}, and \textbf{My Crypto Heroes}~\cite{mycryptoheroes} are games where players control and upgrade their NFTs and battle with the NFTs of others. \textbf{CryptoKitties} and \textbf{Sorare}~\cite{sorare} are games where players cannot change the attributes of their NFTs but can use the NFTs to play functions or battles.  Different from all the above games, \textbf{Axie Infinity} uses Ronin, a side chain of Ethereum. 
A more detailed description of these games is in the Appendix~\ref{sec:gameintro}.

In this table, the number of transactions indicates how active the game interacts with the blockchain, and the total supply of NFTs indicates the rarity of the NFTs. The number of NFT holders reflects the interests of the players who participate in the game. Since if players would not like to stay in the game anymore, they are more willing to sell their NFTs.  
\section{Play-to-earn (P2E) Analysis}
\label{sec:p2e}

In this section, we focus on the play-to-earn (P2E) model by analyzing its operating principle and obtaining the real NFT transaction data to examine whether this P2E model is effective (i.e., do players really earn in these games?).

Literally, P2E means that players can earn benefits only through the processing of playing the game without extra effort. Tokens (either NFTs or fungible tokens) are the rewards that attract players to stay or participate in the management of the decentralized ecosystem. However, tokens are not typically given for free. There are different ways where players can earn tokens through their efforts or participation, while the majority of earnings require investment. Based on whether the players are required to invest before earning, we classify the profit-making approaches into two categories: earn-with-payment and earn-without-payment.

For earn-without-payment, we find that  players can earn limited incomes in tokens without any money out-of-pocket in the NFT games, and the amount of income is limited, unstable, and unpredictable. When compared to the earn-with-payment, its impact on profit is trivial. We provide a detailed analysis of earn-without-payment in Appendix~\ref{sec:earn-wo-payment} and focus on the earn-with-payment analysis in the following. 

\subsection{Few top wallets control significant portions of NFTs; majority of wallets only own one NFT and do not actively trade} 

To study the P2E, first, we examine the NFT distribution among players. That is, to understand how NFTs are distributed among players and how actively they trade NFTs.
For this purpose, we focus on the user transaction activities in games on the Ethereum. In this analysis, we assume a single wallet representing a user\footnote{We will use wallet and player/user interchangeably.}. Given that a user may own multiple wallets, our analysis reflects the lower bounds regarding the NFTs owned by individual users. 

\begin{figure*}[!h]
\captionsetup[figure]{font=small}
\centering
\begin{minipage}{.32\textwidth}
\centering
\captionsetup{font=small}
\includegraphics[width=.9\linewidth]{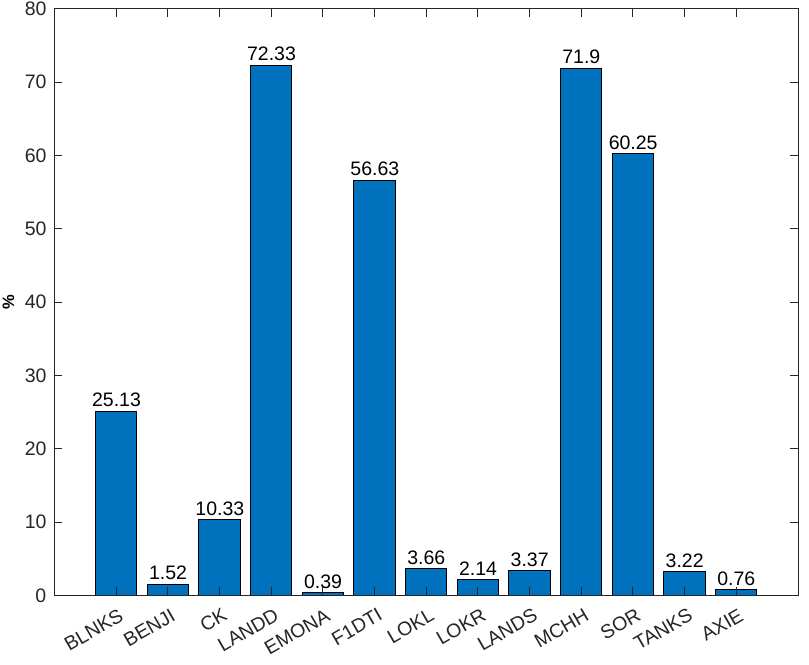}
\vspace{-10pt}
\caption{The largest percentage of NFTs owned by a single wallet \normalfont{(as of 06/30/23)}}
\label{fig:most_hold}
\end{minipage}
\hfill
\begin{minipage}{.32\textwidth}
  \centering
  \captionsetup{font=small}
\includegraphics[width=.9\linewidth]{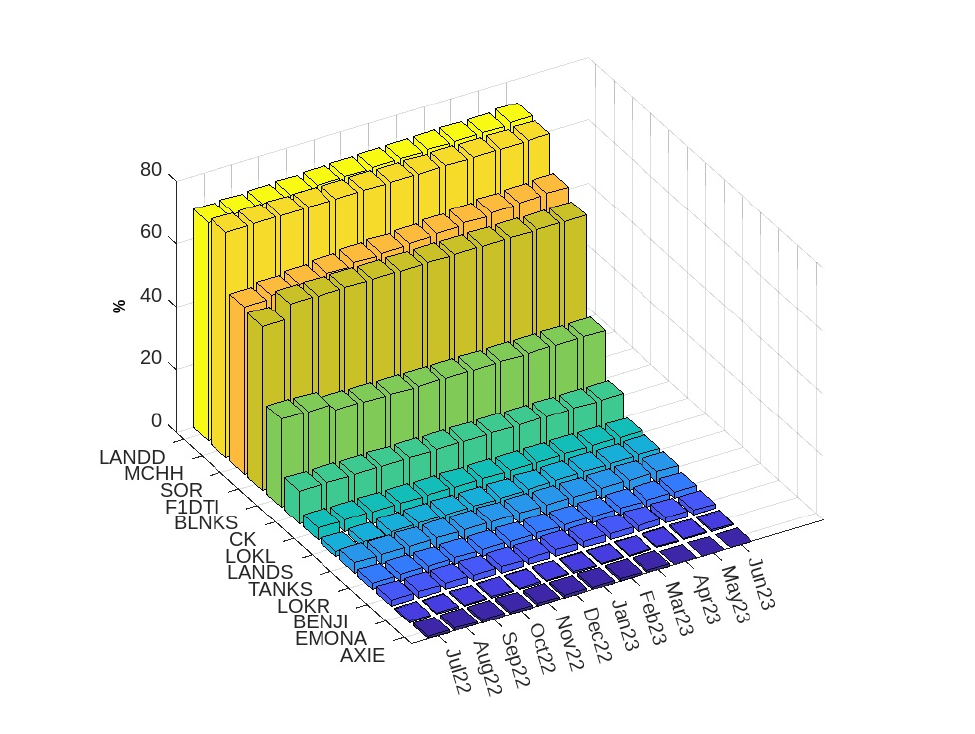}
\vspace{-10pt}
\caption{The largest percentage of NFTs owned by a single wallet \normalfont{(change over time)}}
%address in the Ethereum NFT games}
\label{fig:most_hold_3d}
\end{minipage}
\hfill
\begin{minipage}{.32\textwidth}
\centering
\captionsetup{font=small}
\includegraphics[width=.9\linewidth]{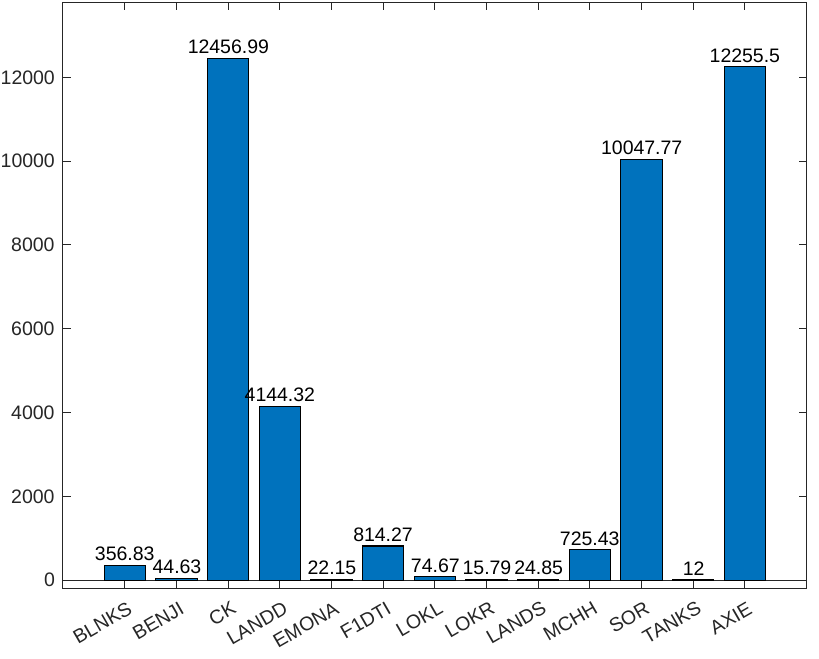}
\vspace{-10pt}
\caption{The ODI in the games \normalfont{(as of 06/30/23)}}
\label{fig:odi_recent}
\end{minipage}
\vspace{-10pt}
\end{figure*}

Figure~\ref{fig:most_hold} shows the largest percentage of NFTs held by a single Ethereum wallet in these NFT games as of 06/30/2023. The higher the percentage, the more concentration of the NFTs, and the stronger the power of this single wallet, because the NFTs contain the voting power in the NFT games. We can see that NFTs of LANDD (Decentraland), MCHH (My Crypto Heroes), SOR (Sorare), and F1TDI (F1 Delta Time) are highly concentrated as one single wallet can own more than 50\% of the NFTs. In LANDS (Sandbox), the largest wallet owns less than 50\%, but still over 1/3 of the total NFTs. 

Figure~\ref{fig:most_hold_3d} further shows how the largest percentage owned by a single wallet has changed over time.  Since each game started at different time, we plot from the start time of the latest game.  As shown in Figure~\ref{fig:most_hold_3d}, the largest percentage of NFTs owned by a single wallet in each game has shown some small variations, with an increasing trend, over time, but overall, it remains stable.

The largest percentage of NFTs owned by a single wallet does not indicate the relative dominant power over other wallets.
To more quantitatively characterize the relative dominance of this wallet, we define a metric, called ownership dominance index (ODI), as follows:
$$ODI=\frac{Largest\_Num\_of\_NFTs}{Average\_Num\_of\_NFTs},$$ 
where the $Largest\_Num\_of\_NFTs$ represents the largest number of NFTs owned by a single wallet, while the $Average\_Num\_of\_NFTs$ represents the average of NFTs owned by all wallets.

\begin{figure*}[!h]
\captionsetup[figure]{font=small}
\centering
\begin{minipage}{.32\textwidth}
\centering
\captionsetup{font=small}
\includegraphics[width=.9\linewidth]{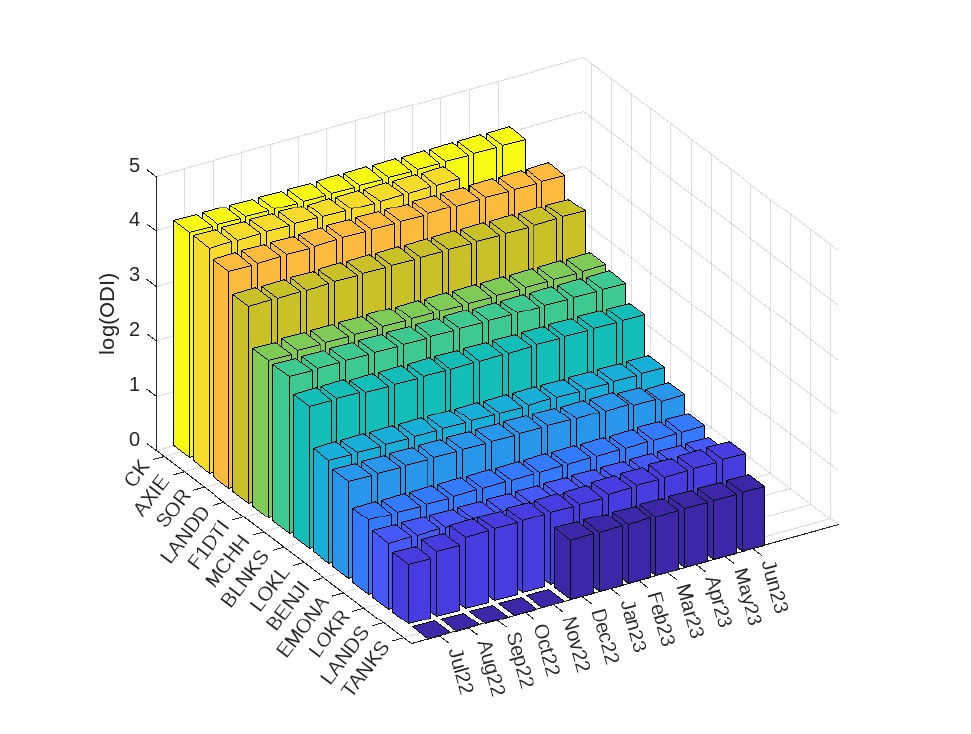}
\vspace{-10pt}
\caption{The ODI in the games \normalfont{(change over time)}}
\label{fig:ODI-3d}
\end{minipage}
\hfill
\begin{minipage}{.32\textwidth}
\centering
\captionsetup{font=small}
\includegraphics[width=.9\linewidth]{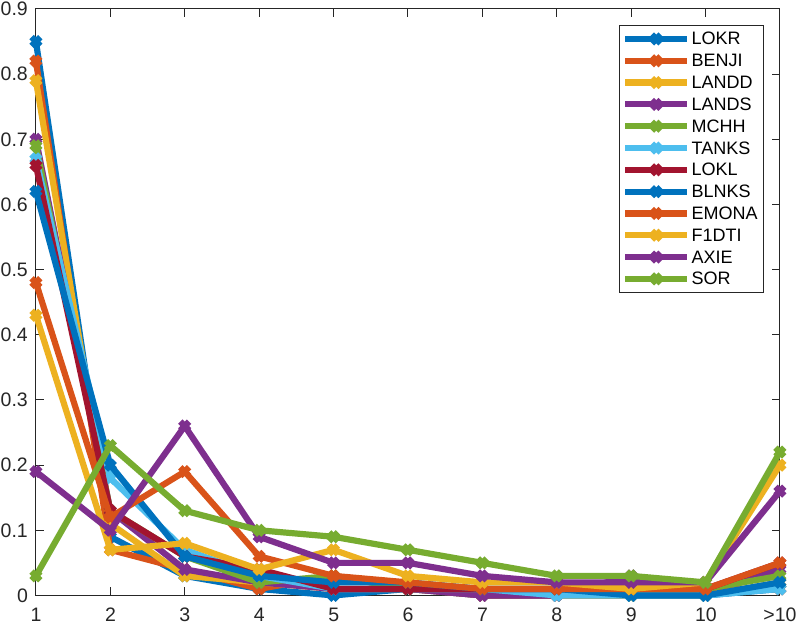}
\vspace{-10pt}
\caption{Wallets that owns one or more NFTs}
\label{fig:num_hold_absolute}
\end{minipage}
\hfill
\begin{minipage}{.32\textwidth}
\centering
\captionsetup{font=small}
\includegraphics[width=.9\linewidth]{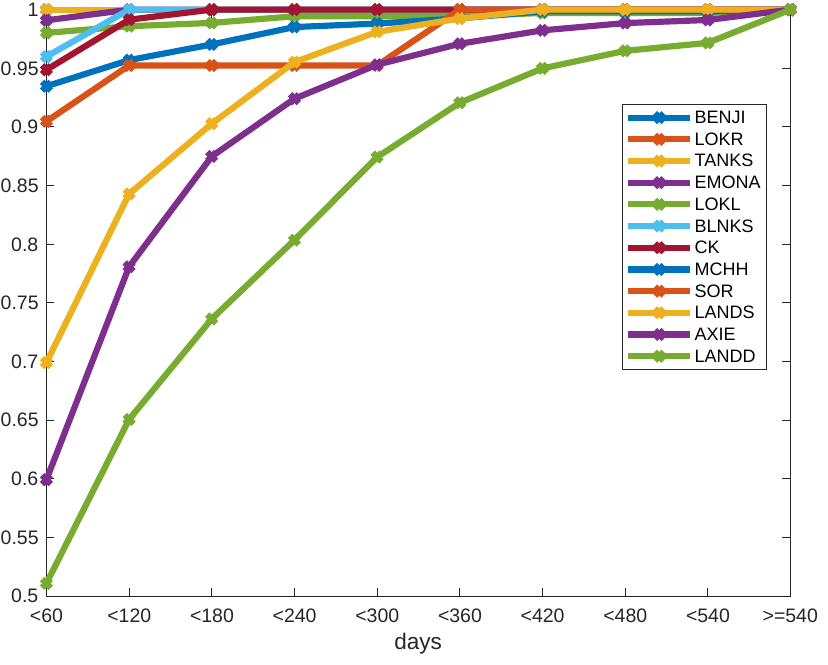}
\vspace{-10pt}
\caption{Time interval (days) between the two purchases by the same wallet}
\label{fig:interval}
\end{minipage}
\vspace{-5pt}
\end{figure*}
Figure~\ref{fig:odi_recent} shows the ownership dominance index (ODI) for each game as of 06/30/2023, while Figure~\ref{fig:ODI-3d} (the y-axis is in log scale) shows how the ODI changes over time.
As shown in the figures, the ODI for those games range from 10 to over 10000, a wide range. Even for BLNKS and EMONA where the largest percentage owned by a single wallet is relatively small, 1.52\%, and 0.39\% as shown in Figure~\ref{fig:most_hold}, respectively, their ODI values show that single wallet still has over 300 and 20 more times NFTs, respectively, than an average wallet owns in these games. Even worse, we can see from Figure~\ref{fig:ODI-3d} that the ODI for most games has been kept increasing over time. This shows a trend of NFT ownership concentration and potentially is not good for NFT trades and circulation among players. 

By excluding the wallet owning the largest percentage in each game, we also calculate the ODI of the wallet that owns the second largest percentage of the NFTs over the rest of the wallets. The result shows a similar trend, with the ODI falling in the range from 12 to 77. While the ODI values are much smaller than before, the relative dominance between the second largest wallet and the rest wallets is still significant.

\begin{table}
\footnotesize
\centering
\caption{Summary of Ethereum based NFT Games. \normalfont{Top wallets own the majority of NFTs in each game.}}
\label{tab:games_eth_hold_top}
\begin{tabular}{c|c|c|c|c}
NFT& Top 1 holds& Top 10 holds& Top 25 holds& Top 50 holds\\ \hline
BLNKS& 25.13\%& 42.03\%& 47.23\%& 52.05\%\\\hline
BENJI & 1.52\%& 8.74\%& 14.06\%& 19.34\%\\\hline
CK& 10.33\%& 31.4\%& 39.11\%& 45.72\%\\\hline
LANDD& 72.33\%& 81.77\%& 82.25\%& 83.71\%\\\hline
EMONA& 0.39\%& 2.62\%& 4.43\%& 6.24\%\\\hline
LOKL& 3.66\%& 23.57\%& 35.16\%& 44.66\%\\\hline
LOKR& 2.14\%& 11.59\%& 19.08\%& 26.97\%\\\hline
MCHH& 71.90\%& 90.89\%& 91.9\%& 92.92\%\\\hline
LANDS&3.37\%&50.15\%&56.39\%&60.15\%\\\hline
SOR& 60.25\%& 64.42\%& 66.45\%& 68.53\%\\\hline
TANKS& 3.22\%& 15.58\%& 26.32\%& 38.58\%\\\hline
F1DTI& 56.63\%& 67.93\%& 73.78\%& 78.33\%\\\hline
AXIE& 0.76\%&6.25\%&6.89\%&7.56\%\\
\end{tabular}
\end{table}

The high ownership dominance result drives us to study how many NFTs an average wallet may own. We plot in 
Figure~\ref{fig:num_hold_absolute} the distribution of wallets that own one or more NFTs in each game. As we can observe, for most of the games, the majority of the wallets only own one NFT. For example, Sandbox and Decentraland are very similar in genres and content, and we observe very similar patterns where over 80\% of the LAND owners have no more than two LANDs simultaneously.  Table~\ref{tab:games_eth_hold_top} further summarizes the NFTs owned by top wallets. Overall, we can observe that top 50 wallets in each game own more than 1/3 of the total NFTs in 9 games.  

\emph{Combining the findings from  Figure~\ref{fig:most_hold} to Table~\ref{tab:games_eth_hold_top}, we can find that currently, the ownership dominance in these NFT games is very high, which is detrimental to a thriving NFT game marketplace.
Normally, we understand that the game developer's team keeps a large portion of the NFTs for their profit. However, as we have observed here, how much the developer's team should control deserves careful deliberations. While we cannot make a suggestion based on the current data, we believe that the high concentration of NFT ownership will hinder the circulation of NFTs, and eventually, the economic model itself. 
}

\begin{table}
\footnotesize
\centering
\caption{The average time interval (days) between the two purchases by the same wallet}
\label{tab:interval}
\begin{tabular}{c|c|c}
NFT & Average Time Interval (days)& SD of Time Interval\\ \hline
BLNKS & 8.6 & 22.79\\\hline
BENJI & 0.004& 0.06\\\hline
CK & 7.3 & 21.42\\\hline
LANDD & 121.33 & 151.23\\\hline
EMONA & 2.28 & 11.41\\\hline
LOKL & 7.98 & 42.58\\\hline
LOKR & 1.44 & 4.28\\\hline
MCHH & 15.74 & 55.61\\\hline
LANDS& 71.74 &84.11\\\hline
SOR & 28.9 & 76.97\\\hline
TANKS & 0.03 & 0.23\\\hline
AXIE & 72.29 & 92.98\\
\end{tabular}
\end{table}
 
Given the high dominance of single wallet, we conjecture that such a trend may also adversely  affect how players actively conduct transactions. 
Thus, we further look into how active users are conducting the NFT trades. To capture the activeness level of players, we focus on how frequently users buy and sell NFTs. As they show similar pattern, we present the result on purchases. 

Figure~\ref{fig:interval} shows the CDF of inter-purchase time between two purchases from the same wallet. The inter-purchase time is the days between two consecutive purchases by the same player. As we can observe, for example, in  Decentraland, 61.55\% of the NFT purchases happened within 100 days after the last purchase. The median between two LAND purchases is 55 days, and the average is 121.33 days. For Sandbox, 59.54\% of the next purchase happened within 30 days of the last purchase, with a median of 13 days, and an average of 53.49 days.
 
For Axie Infinity, 59.92\% of the next purchases happened within 60 days of the last AXIE purchase, with an average of 72.29 days. Overall, more than half of the purchases happened in 60 days for all games. 

Table~\ref{tab:interval} further shows the average days between two purchases with the standard deviation. As the result shows, except for BENJI and TANK, in all other games, the subsequent purchase of the same NFT always happens some days later, up to 121 days for DecentraLand. Such results indicate that NFT trades among players are not very active. Most players would just hold one or several NFTs for a long time without trading. Note that we did not list F1TDI because it has been shutdown. There were a lot of withdraw trades affected by that, which significantly disturbed the inter-purchase interval analysis. Thus we excluded that here and subsequent analysis where appropriate.  

\begin{figure*}[!h]
\captionsetup[figure]{font=small}
\centering
\begin{minipage}{.32\textwidth}
\centering
\captionsetup{font=small}
\includegraphics[width=0.9\textwidth]{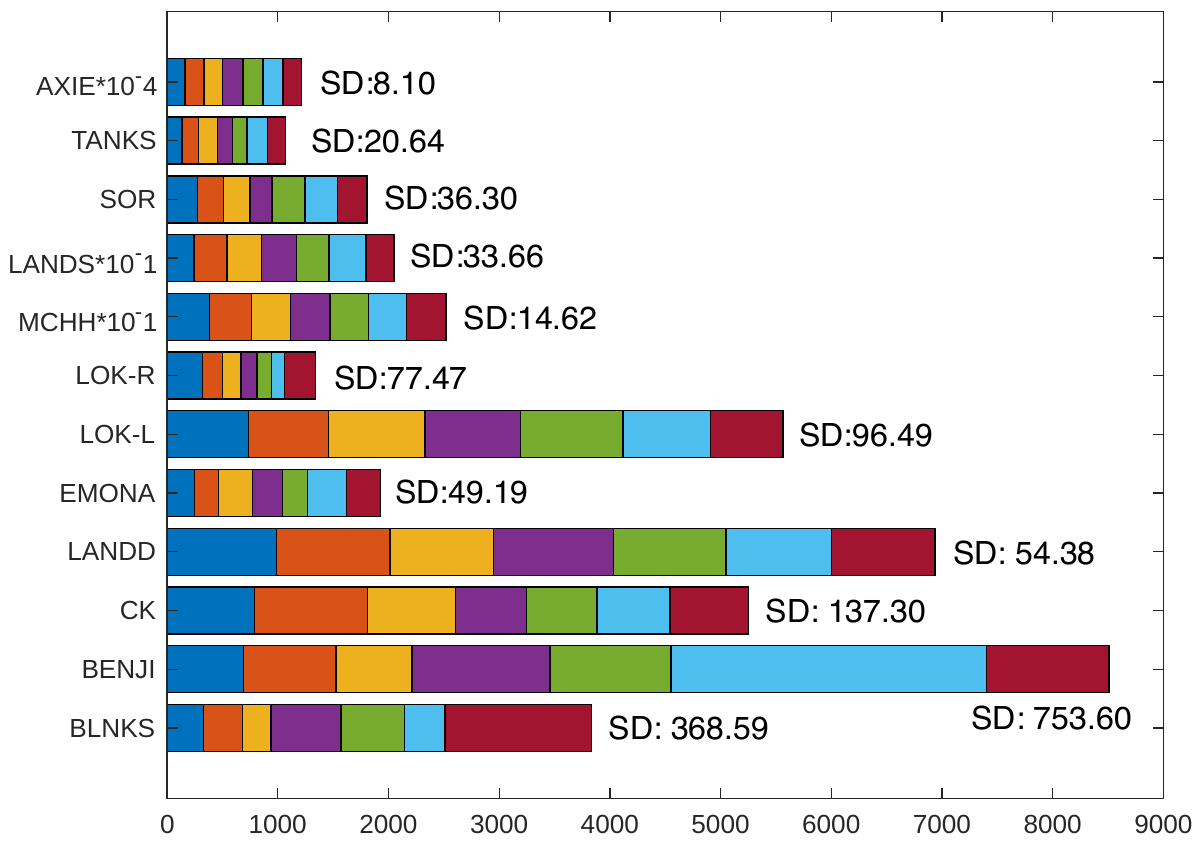}
\caption{NFT trade average: 7 days of a week}
\label{fig:trade_day_of_week}
\end{minipage}
\hfill
\begin{minipage}{.32\textwidth}
\centering
\captionsetup{font=small}
\includegraphics[width=0.9\textwidth]{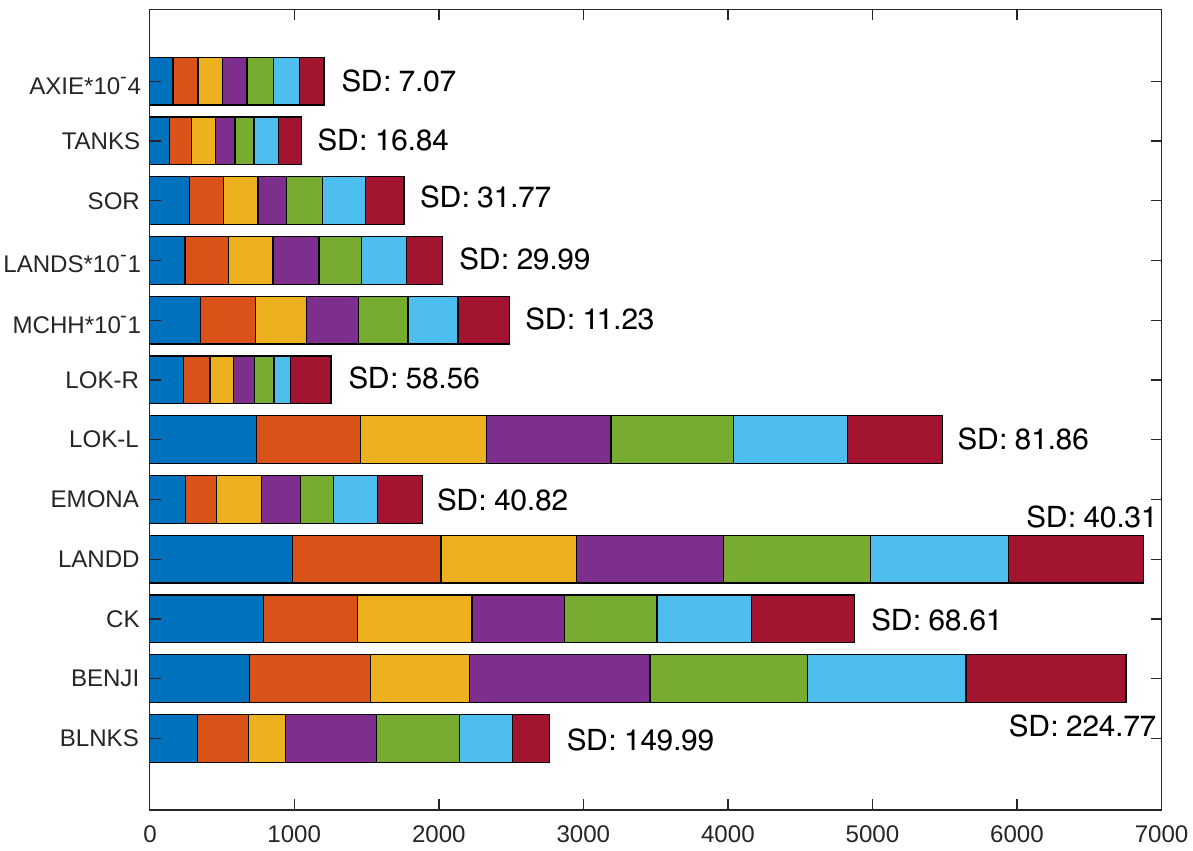}
\caption{NFT trade average: 7 days of a week \normalfont{(excluding the "Big Day" trades)}}
\label{fig:trade_day_of_week_exclude}
\end{minipage}
\hfill
\begin{minipage}{.32\textwidth}
\centering
\captionsetup{font=small}
\includegraphics[width=.9\linewidth]{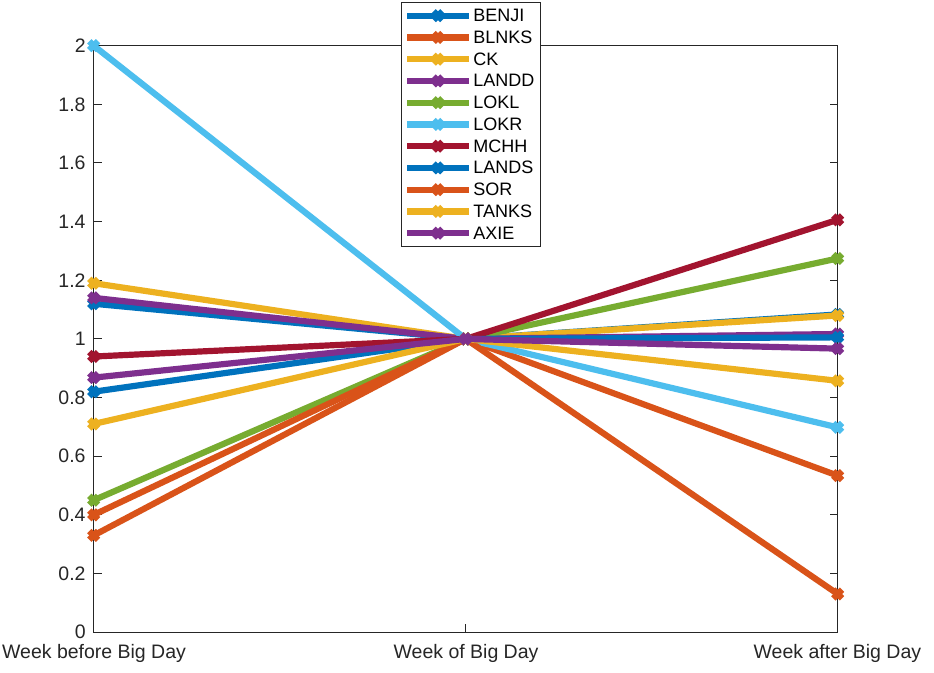}
\caption{Price changes before and after the ``Big Day"s}
\label{fig:big_day_week_price}
\end{minipage}
\vspace{-10pt}
\end{figure*}
\vspace{0.1in}
\noindent\textbf{Takeaway \#1:}\emph{The NFT distribution analysis shows that currently the NFT distribution among players is highly concentrated, and the large number of players who own a small number of NFTs are not actively conducting transactions. This may be detrimental to the NFT games ecosystem in the long run. }

\subsection{Boosting effect of promotion events varies game by game and does not sustain}

Having studied the NFT distribution among players, next, we investigate the trades, as trades are directly contributing to the profit or loss of the players. 

As an NFT game player, the chance to earn profits in terms of investment seems more straightforward and realistic because players are expected to at least interact with NFTs and/or other tokens when playing the game. NFT games often provide a detailed guide and policy description on how to earn and what is expected to earn through the investment. The players can not only invest in NFTs but also invest utility tokens or governance tokens. A player can invest the tokens and earn profits through the following approaches: (1) NFT Trading, (2)  NFT Functioning, (3) Utility Token Trading, (4) Utility Token Functioning, and (5) Utility Token Staking.

Next, we discuss the earning based on the real-world NFT trading data retrieved from the blockchain.

\subsubsection{NFT Trading}\hfill

To study the NFT trades, we first summarize day-of-week trade average through the timestamp with the UTC timezone in Figure~\ref{fig:trade_day_of_week}. Figure~\ref{fig:trade_day_of_week} shows that different NFT games have trading peaks on different days, however, there is no consistent pattern (e.g., always peaks on Tuesday) that we could observe from the analysis, except that for each game, there is a day within which the average number of trades is larger than other days. 

As NFT trades currently are not constrained by the time or location, we wonder if the larger number of trades on a day may be caused by some promotion events. To investigate this, we conduct the following analysis in two steps. 
First, for each game, we find out the day with the largest number of trades, which we call "Big Day" in Table~\ref{tab:day-of-week}.
The trades on these "Big Day"s are significant in terms of the number of trades. 
When we exclude those "Big Day" trades, shown in Figure~\ref{fig:trade_day_of_week_exclude}, we find that the standard deviation of the daily trade counts reduced remarkably and the daily trade numbers remain stable. This suggests that the "Big Day" does have a "big" effect on these games. 

Second, since the "Big Day" matters in trading counts, we want to find out what led a day to "Big Day". It is intuitive that the NFT platform or the application providers have the intention to  have promotion events to promote the NFT trading, such as marketing events. In addition, other events, such as those triggered by some ad-hoc tweets, can also lead to a peak of NFT trades.   
Therefore, for all these games that we have identified the day with the largest trade amount, i.e., "Big Day", we tried to confirm with other information resources (e.g., news media) that some promotion event did happen for that game. 
By checking each of these "Big Day"s, we found that they belong to the following categories, as shown in Table~\ref{tab:day-of-week}. 
\begin{enumerate}
        \item Marketing. We can see among the 12  ``Big Day"s, 6 of them are within 40 days from the earliest NFT trade. Although the marketing period  vary from several weeks to months among different games, there is no doubt that marketing will be extremely active around the game release date. The earliest NFT trade may not be not identical to the game launch date but it would be close to the game release date. For instance, The Sandbox published a video named "What Can I Do With My LAND" on Youtube on 2/9/2022, \ie exactly on the "Big Day".
        \item Big News. Blankos Block Party announced events with Amazon Prime Gaming on 3/15/2022, \ie 4 days before the "Big Day". Samsung launches metaverse store in Decentraland on 1/7/2022, \ie exactly on the "Big Day". Major video game publisher Ubisoft announced a partnership with The Sandbox on 2/8/2022, \ie 1 day before the "Big Day". Sorare announced the MLB auction on 7/24/2022, which is 10 days before the "Big Day".
        \item Social Media. A gaming Youtuber published a video introducing Blankos Block Party on 3/17/2022. \ie 2 days before the "Big Day". Another Youtuber who has 322.1k subscribers named "RJ Jacinto" published his first video playing Axie Infinity on 11/5/2021, which is exactly the "Big Day". The video has 26k views. 
\end{enumerate}

\begin{table*}
\centering
\noindent
\footnotesize
\caption{"Big Day"  (with the largest number of daily trades) of Each Game}
\label{tab:day-of-week}
\begin{tabularx}{\textwidth}{X|X|X|X|X|X|X|X|X|X|X|X|X}
 & BLNKS & BENJI & CK & \makecell{LANDD} & EMONA & LOKL & LOKR & MCHH & \makecell{LANDS} & SOR & TANKS & AXIE \\\hline
Big Day & 3/19/22 & 2/18/22& 11/21/22 & 1/7/22 & 10/17/18 & 7/8/20 & 7/11/20 & 3/31/19 & 2/9/22 & 8/4/22 & 6/23/23 & 11/5/21 \\\hline
Day of Week & Sat. & Fri. & Mon. & Fri. & Wed. &  Wed. & Sat. & Sun. & Wed. & Thurs. & Fri. & Fri. \\\hline
Days since the 1st trade & 39 & 1 & 99 & 466 & 19 & 33 & 7 & 168 & 13 & 382 & 181 & 191 \\\hline
{\bf \# of Trades} & {\bf 1,068} & {\bf 1,755} & {\bf 376 } & {\bf 62} & {\bf 41} & {\bf 80} & {\bf 86} & {\bf 301} & {\bf 228} & {\bf 48} & {\bf 16} & {\bf 108,207} \\\hline
{\bf \% of Total Trades} & {\bf 27.85\%} & {\bf 20.63\%} & {\bf 7.16\%} & {\bf 0.89\%} & {\bf 2.13\%} & {\bf 1.44\%} & {\bf 6.42\%} & {\bf 1.20\%} & {\bf 1.11\%} & {\bf 2.66\% }& {\bf 1.50\% }& {\bf 0.89\% }\\\hline
Event & Social Media, News, Marketing & Marketing & Marketing & News & Marketing & Marketing & Marketing & Social Media & Marketing & Marketing, Social Media & Social Media & Social Media\\\hline
Price re-stable interval & 7&5&3&2&1&14&5&4&6&3&4&2\\\hline
Price re-stable rate & 1.01&0.96&1.01&1.07&0.84&1.15&0.81&1.1&0.85&1.09&1.06&1.03\\
\end{tabularx}
\end{table*}

Note that here we only presented the trades on the identified "Big Day" so far in Table~\ref{tab:day-of-week}. Ideally, we want to precisely quantify the trade counts due to such promotion events. However, in practice, we find it is hard, if not possible, to precisely know how long the effect of such events diminish over time. That is why we name it as "Big Day", in order to differentiate it from the entire promotion events. Even so, by just taking out the trades caused by these single "Big Day"s, we found that the trades weekly of these games largely remain stable.  
Furthermore, there may exist other promotion events during our investigation period. However, they cannot be manually confirmed by us via other public sources as we did for the "Big Day", and thus we did not report here. 

Since the promotion events may not only boost the number of transactions, but also the price, we also study the price change around "Big Day"s. Figure~\ref{fig:big_day_week_price}  presents the price change rate before and after the identified "Big Day"s. 
Since different NFTs are priced very differently, in this figure, we set the average price of the NFT on the week of the Big Day as the baseline and measure the average price one week before and after  the "Big Day"s. 
Interestingly, we find that (1) the promotion events do show some boosting effect for NFT price as well for the majority of these games, but also show some counter effect for several games, suggesting more research on promotion strategies is needed, 
 
(2)  more importantly, the NFT price increase and the price decrease are largely symmetric around the "Big Day"s, indicating that the promotion effect does not sustain.  
We further set the average trade price on ``Big Day" as the baseline, and calculate the cumulative average trade price for each game, \ie the average trade price of 1-day, 2-day, 3-day trades after the ``Big Day", and calculate its ratio against the ``Big Day" price, shown in Table~\ref{tab:day-of-week} as {\em Price re-stable rate} when the closest is reached. Accordingly, {\em Price re-stable interval} shows the time (days) it takes. 
As we can observe from the result, the longest time of the promotion effect is about two weeks, while most of the promotion effect only lasts no more than one week.

\subsubsection{NFT Functioning}\hfill

NFTs are essentially game assets or props in the games. Thus, when the owner performs the functionality of the NFT in the game, NFT can earn tokens for its owner. For instance, the LAND owner in Decentraland~\cite{decentraland} can rent her/his LAND to others to hold events and charge the rent. The daily rent range is between 1 and 1,000 MANA. Axie Infinity~\cite{Axie} players can make a three-axie team and let the axies battle with the axies of other players. The battle is under the player vs. player (PvP) mode. Players earn 5 SLP for every winning in the PvP mode.

However, earnings gained from these approaches require the players to purchase NFTs first. Compared to the average price of NFTs \ie \$11194 in Decentraland and \$203.43 in Axie Infinity, the amount of earnings in MANA (amount starts from 1) and SLP ( around 0.027 USD per SLP) is extremely low, and  can be  ignored.

\subsubsection{Utility Token trading, functioning, staking}\hfill

We provide a detailed analysis of utility tokens in Appendix \ref{sec:utility}. We found that their contribution to the profit of the players is trivial.

\vspace{0.1in}
\noindent\textbf{Takeaway \#2:} \emph{From the analysis of the real-world NFT trades, we find that  while
the promotion events can promote the trades amount and prices for NFT games, their effect varies game by game, and such effects do not sustain for a long time, in terms of both the boosted traded amount and the boosted NFT price. 
}
\subsection{Majority of the players did not earn}
\begin{table*}
\centering
\caption{\label{tab:profits} Summary of NFT Trade Profits}
\begin{tabular}{c|c|c|c|c|c|c||c}
NFT & Trades & Traded NFT & \makecell{Average Trade \\Price (\$)} & \makecell{Average Profit (\$)} & \makecell{Median \\of Profit (\$)} & \makecell{Mode of\\ Profit (\$) }& \makecell{Developer's Team\\ Average (Selling) Profit (\$)}\\\hline
BENJI  & 8,298 & 5,000 & 573.26 & \textcolor{blue}{-0.16} & 0 & 0 & N/A \\\hline
 BLNKS & 3,835 & 2,317 & 361.00 & 19.22 & \textcolor{purple}{-2.06} & 0 & 282.81\\\hline
 CK & 5,250 & 4,883 & 35.57 & 78.64 & 0.74 & \textcolor{red}{-1.26}  &  22.93\\\hline
 LANDD & 4,505 & 3,103 & 12,295.24 & \textcolor{blue}{-213.19} & 295.78 & 325.73  & 12172.79\\\hline
 EMONA & 1,928 & 1,700 & 106.96 & \textcolor{blue}{-22.51} & \textcolor{purple}{-0.24}& N/A  & 52.52\\\hline
 LOKL & 5,565 & 2,906 & 1315.28 & 721.38  & 146.11& 174.38 & 933.46\\\hline
 LOKR & 1,341 & 1,218 & 32.01 & 82.18  & 21.61 & \textcolor{red}{-83.34}& 13.54\\\hline
MCHH & 25,183 & 13,854 & 120.04 & \textcolor{blue}{-3.89} & \textcolor{purple}{-0.40}& \textcolor{red}{-45.37} & 114.99\\\hline
 SOR & 1,807 & 1,613 & 116.27 & \textcolor{blue}{-33.40} & \textcolor{purple}{-1.18} & \textcolor{red}{-5.21}  & 112.92\\\hline
 TANKS & 1,067 & 909 & 473.61 & 207.07 & 117.94 & N/A  & 465.82\\\hline
LANDS & 20,480 & 14,120 & 6,148.78 & 156.24 & \textcolor{purple}{-49.04} & \textcolor{red}{-1.92}  & 6732.83\\\hline
 AXIE & 12,168,031 & 7,334,880 & 203.43 & 196.93 & 145.33 & 0  & 331.65\\
\end{tabular}
\end{table*}
\vspace{-5pt}
 
After studying the NFT distribution and NFT trade distribution, in this section, we focus on whether the players actually make profits via these trades. 
Table~\ref{tab:profits} summarizes the player NFT trade profits of these games. Again, we excluded the shutdown game F1TDI in the profit analysis. 
In Table~\ref{tab:profits}, {\em Trades} shows the number of NFT transfer transactions, including the initial sales from the game-developing team (\ie the first-hand transactions) and the sales between the players (\ie all transactions except the first-hand sale, which is from the developer's team). {\em Traded NFTs} represents the number of unique NFTs involved in the transactions we collected. 
The closer the value of {\em Traded NFT} and {\em Trades} is, the larger the proportion of NFT first-hand transactions is, the worse the circulation of NFT will be. 
Intuitively, selling and purchasing lead to circulation, which is directly related to the profitability of NFT transactions.  
All the profits are calculated as the difference between two consecutive trade transactions of the same NFT (\ie based on the token ID). The {\em Average Profit} is the average of the profit per trade (a negative profit indicates a loss). We also calculate the average trade price in the table~\ref{tab:profits}. 
Since conducting transactions on the blockchain would consume resources (\eg gas on Ethereum), which is called transaction fees, for a transaction where the profit is 0, we still count it as a loss because the player needs to pay for the transaction fee when she/he traded the NFT. 
In the table, {\em Median} refers to the median value of player's profits, while {\em Mode} refers to the most frequently seen profit value that appears in the trade profits of the game.

From Table~\ref{tab:profits}, we can observe that 6 games have a negative profit, in terms of {\em Average Profit}, indicating that players in these games do not earn on average. Since the average profit could be affected by the extreme profit values, we also report the {\em Median} of the player's profit. As shown in the table, even though the {\em Average} profit in Blankos Block Party is about \$19, the {\em Median} of the profit is negative, \$-2.06, indicating half of the players in this game lost money. Similarly, in Cryptokitties, the {\em Average} profit of the trade is \$78, but the {\em Median} is \$0.79, close to 0,  indicating half of the players barely made any profit. 
By further looking into the {\em Mode} in the table, we find that among the remaining games with a positive {\em Average} profit, two more have a negative {\em Mode} value, indicating the most frequently profit made by players in these games is negative.  Putting all together, we believe that on average players in 9 games did not earn. 

On the other hand, the last column in Table~\ref{tab:profits} shows the {\em Average} profit of the trades conducted by the developer's team. Note for their trade, we consider their cost is 0, since we do not know exactly how much cost they have paid. Note that we cannot calculate in the case of BENJI and it uses ERC-1155 with same token ID in the collection. This is the game specific setup because BENJI is just a qualification used to participate the game and receive incentive.
Thus, their trades always make money, and shows a upper bound of the real profit. 
Nevertheless, compared to the profit situation of the ordinary players, the result suggests that the majority profits of the games went to the developer's team. 
\begin{figure}
\centering
\includegraphics[width=0.46\textwidth]{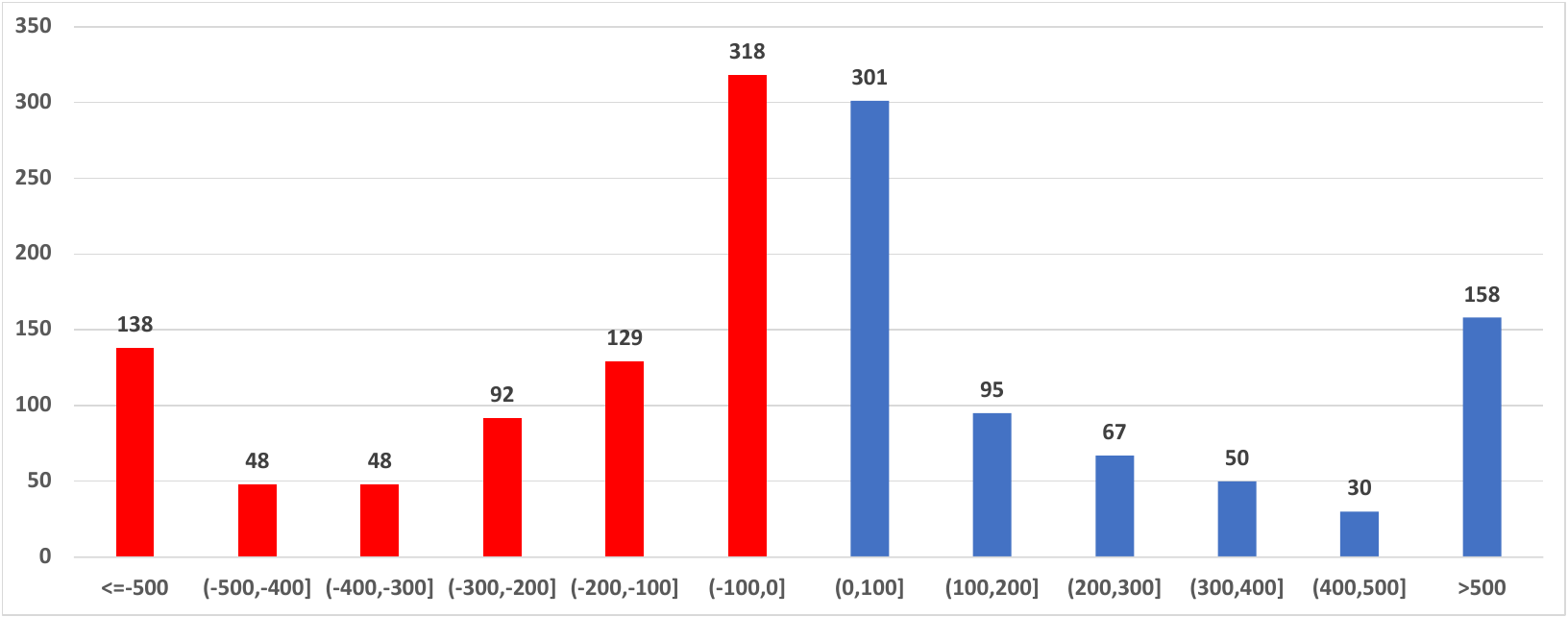}
\caption{The profit distribution for BLNKS}
\label{fig:blnks-pro-dis}
% \vspace{-10pt}
\end{figure}
We further calculate the profit of each trade in each game.  Figure~\ref{fig:blnks-pro-dis} shows an example of BLNKS (recall it has a positive average profit, but a negative median profit).  As we can observe from the figure, most of the profits are in the range between \$-100 and \$100, with more on the negative sides than on the positive side. Most of other games have a similar pattern. 

\noindent\textbf{Takeaway \#3:} 
\emph{The profit analysis shows that the majority of profits are earned by the game developer's team at the first transaction of the NFT, while the majority of the players  can hardly earn profit in these games.} 

\vspace{0.1in}
\noindent\textbf{In summary, by looking into the NFT distribution among players, effect of promotion events, and profit siutation of players, we find that, currently, the play-to-earn model largely does not work well in these games. Players should take the risks into consideration when participating in these games.}

\begin{table*}
\centering
\small
\caption{Profit model using linear regression}
\label{tbl:profitmodel-1}
\begin{tabular}{c|c|c|c|c|c|c||c}
Game & $\alpha$ & \makecell{Average Trade \\Price (\$)} & $\beta$ & Circulation rate & $\gamma$ &Mean Profit& Estimated Profit  \\\hline
BENJI&18.96&573.26&-0.38&0.397445168&500&-0.16&-0.1162\\\hline
BLNKS & 22.3 & 361 &-0.56 & 0.395827901 & 500 &19.22& 18.0540 \\\hline
CK&3.557&35.57&-0.002&0.069904762&100&78.64&10.4763\\\hline
LANDD & 271 & 12295.24 & -0.06 & 0.311209767 & 800 &-213.19&-217.7466   \\\hline
EMONA&1.06&106.96&-0.008&0.118257261&100&-22.51&12.03\\\hline
LOKL & 85.49 & 1315.28 & -0.03 &0.477807727 & 1400 & 721.38&714.9656 \\\hline
LOKR&0.1601&32.01&-0.002&0.091722595&100&82.18&9.2683\\\hline
MCHH&-12&120.04&-1.00&0.449866974&900&-3.89&272.8403\\\hline
SOR&11.62&116.27&-0.46&0.107360266&100&-33.4&-31.1282\\\hline
TANKS&30&473.61&-0.004&0.148078725&100&207.07&42.9134\\\hline
LANDS & 30.7439 & 6148.78 & -0.058 & 0.310546875 & 900&-49.04&-46.3932 \\\hline
AXIE &1.0172 & 203.43 & -1.0 & 0.397200747 & 1000 &196.93& 194.7879 \\
\end{tabular}
\end{table*}

\section{Model for current P2E in the NFT games}
\label{sec:extend}
In an NFT game ecosystem, there are two types of players and we assume they are rational in participating in the game. One type of player is the ordinary player (players) in the game. We assume that there are $N$ players in total. The other type of player is the game designer or the platform. We use $P$ to denote the platform. The time interval is discretized as denoted as $t = 0, 1, 2, \ldots, T$, where the $t$th time slot denotes the time interval $[t, t + 1)$. $T$ is the expected lifespan of the game. Suppose there are $K$ NFTs in total for the game and every NFT has the same price at the same time interval and we denote that as $v_t$. 
The platform can set the initial price of the NFTs when the game is released and the price is denoted as $v_0$. At each time interval $t$, the platform can choose to release $k_t$ NFTs. The price of the released NFTs is the same as the price of NFTs that are traded among the players. So, we have $\sum_{n=0}^{T}k_t = K$. At each time interval $t$, players can either buy NFTs from the platform or other players, sell the NFTs she/he already owns, or hold the NFTs and participate in the game. The NFTs have functionalities in the game, when players hold the NFTs and use them to participate in the game, players can gain some benefits, we use $f$ to represent the benefits the NFT owner can earn for using a single NFT to participate in the game during a single time interval. $f$ shows the functionality that NFTs performed in the game and the quality of the game. Thus, $f$ depends on the platform and is positively correlated with the cost of the game. Once the game has been released, $f$ remains the same value. In addition, the market estimates the price of NFTs at the next time interval $t+1$ and we denote that with $v'_t$. 

However, the difference between the estimated price $v'_t$ and the real-world price $v_{t+1}$ commonly exists. As discussed in Section 4.2, the price of NFTs is also significantly affected by external events and some of the events cannot be expected even by the platform. We use $e_t$ to denote the price movement caused by external events. Thus, $v_{t+1}$ can be calculated as $v_{t+1} = v'_t + e_t$. The cost of the platform developing the game is denoted as $C_0 K f$, where $C_0$ is a game-dependent constant value. 

At the beginning of the game, the platform has spent the cost for developing the game, and set $v_0, k_0$, and $f$ for the game. The platform can release and sell NFTs at any time interval $t$. Players can buy NFTs from the platform or from other players. During each time interval, NFT owners can earn $f$ for using the NFT to participate in the game. Players also gain profits by selling the NFTs. Players and the platform can only estimate the NFT price for the next time interval based on the current price and have no future knowledge. The cumulative profits for the platform at time $t$, i.e., $R_{platform}(t)$ is calculated by the following formula:
\begin{align}
R_{platform}(t) = \sum_{n=0}^{t} v_n \times k_n - C_0 K f 
\end{align}

And the cumulative profits for the players at time $t$, that is, $R_{player}(t)$, is calculated by the following formula:
\begin{align}
R_{player}(t) = \sum_{n=1}^{I} (v_{s_n} - v_{b_n} + f(s_n - b_n)) - \sum_{n=1}^{J}v_{b_n} \notag\\+ \sum_{n=1}^{J}(t - b_n)f,
\end{align}
where $I$ represents the number of NFTs that player has sold and $J$ represents the number of NFTs that player still holds. Thus, $I + J \leq \sum_{n=0}^{t}k_t \leq K$. $s_n$ represents the time interval that the NFT is sold and $b_n$ represents the time interval that the NFT is purchased. $s_n, b_n \in 0, 1, 2, ..., T$ and $t > s_n > b_n$.

For player decisions at time $t > 0$, players buy an NFT if the expected return at $t+1$ is less than the cost, that is, $v'_t + f - v_t > 0$, sell the NFT if $v'_t + f - v_t < 0$ and hold the current NFT if $v'_t + f - v_t = 0$.

For platform decisions at time $t$, the platform releases $\frac{(t+1)K}{T+1} - \sum_{n=0}^{t-1}k_n$ NFT when $v'_t > v_t$; releases $\frac{K}{T+1}$ NFTs when $v'_t = v_t$, and holds the rest when $v'_t < v_t$.

Both the platform and the players would like to maximize profits. The platform aims to increase $v_0$ and decrease $C_0Kf$. The maximum potential profits of the platform are $k_0v_0 + (K-k_0)v_{max} - C_0Kf$, $v_{max}$ is the maximum price for all $t$ and $k_{max} = K - k_0$. There are two ways to satisfy the potential maximum profits: 1. $v_t$ continues to increase and the platform releases NFTs at a stable speed; 2. $v_0$ is the maximum price for all $t$ and the platform releases all $K$ NFTs at the beginning.

The potential maximum profits of the players depend on the price fluctuations, $f$, and the number of the NFTs players own, where $v_T$ is the maximum price of the NFTs. If $f = 0$, the strategy of the players will be the same as that of the platform. Thus, players are more tolerant of the price fluctuation within $f$ and players expect to trade more NFTs for larger profits. 

Let’s assume there is a special game player (distinct from the platform) who mimics the platform's actions once they receive the same number of NFTs. Starting from that time slot, we can treat the game as one involving two platforms (the original platform and the special player), both trying to sell NFTs to the remaining regular players. From this point onward, the platform and the special player will compete for revenue, each lowering the NFT values in an effort to maximize their individual revenue.
As a result of this dynamic, no pure Nash equilibrium can exist, because both the platform and the special player continuously adjust their strategies in response to each other’s actions, preventing the game from settling into a stable state.

From the model above, we conclude that the pure Nash Equilibrium for the whole system is nonexistent. Therefore, the current P2E model is not able to make the current NFT game ecosystem stable if players have a pure strategy. 
\section{Incentive Mechanism}
\label{sec:tla}
Based on the empirical data analysis discussed above, we realize that the effectiveness of P2E in the real world is quite different from the original design intention because the earn-without-payment (\ie incentives) is trivial (see Table~\ref{tbl:earn-1} and Appendix~\ref{sec:earn-wo-payment}). The majority of ordinary players are hardly able to expect to make profits in NFT trades. This has a negative impact in the long run for these games, i.e., the games cannot be sustained because players will sway away from these games. The disconnection identified between the intended design goal of Play-to-Earn (P2E) and the empirical analysis results disclosed through the data analysis accentuates the imperative for a critical overhaul in implementing P2E models in the real world. This imperative is particularly pronounced in incentives, which constitute the core element of the P2E model to attract players but yet appear to be under-studied. 

An intuitive consideration emerges in the form of incentive tokens. In most NFT games, incentives are predominantly confined to utility tokens or cryptocurrencies. Although a subset of games incorporates NFTs as incentives, it is noteworthy that players are often required to make NFT purchases before engaging in gameplay.  Many game developers use utility tokens as incentives to avoid the risks of lowering the primary NFT price. This strategic choice effectively mitigates the potential negative influence on NFT prices. However, it is essential to recognize that this approach also diminishes the likelihood of players, enticed by incentives, delving deeply into the game, developing an intrinsic interest in the gameplay itself, and becoming active participants. This is attributable to the limited utility scenarios of most utility tokens within the game environment. Consequently, players receiving incentives may be unable to engage extensively with the game by consuming utility tokens, leading them to sell the incentives and exit the game.

Employing NFTs as incentive tokens undeniably amplifies players' enthusiasm for the game. Simultaneously, to uphold the value, scarcity, and desirability of the primary NFT, developers may contemplate the creation of novel categories of NFTs endowed with augmented attributes or without in-game functionalities. These newly designed NFTs can serve as incentives whose effects on the price of primary NFTs are limited. This strategic approach not only can bolster player engagement but also can contribute to the sustained value proposition of the overarching NFT ecosystem.

Notably, many NFT games incorporate invitation/referral incentives; however, the amount of these incentives tends to be limited. Invitations and referrals emerge as highly effective mechanisms for attracting new players, contingent upon the attractiveness of the invitation incentive amount. Therefore, enhancing the allure of invitation/referral incentives holds the potential to significantly augment the influx of new players into the game community.

However, in several instances, existing incentive amounts have proven insufficient to captivate active players' participation in the game, further diminishing their appeal to potential players. Consequently, a judicious adjustment of incentive amounts becomes imperative to bolster user engagement among current players and entice prospective players to join the game.

To determine an appropriate incentive amount and explore the impacts of an appropriate incentive mechanism for the NFT games, next we develop an incentive mechanism model for this purpose.

\subsection{Incentive Modeling}

\subsubsection{Players in a P2E game.}

In designing an incentive mechanism, firstly, we consider the possible \emph{players} in a P2E game. There are two types of players and we assume they are rational in playing the game. One type of player is the \emph{game player}. Let us name the game players $1, 2, \ldots, N$. Another type of player is the \emph{game designer} (also called the \emph{platform}) of the P2E game. We assume only one such player represents the platform and use $P$ to denote the game designer (i.e., platform). In addition to these two types of players, we assume that there are $K \ge 1$ gaming assets. Each asset has only one copy. At any time, an asset is owned by either a game player or the platform. The prices of these gaming assets can be changing over time. It is assumed that the prices of an asset $j$ fall into a range $[L(j), U(j)]$.

Secondly, we examine the \emph{strategies} that the players can take in this game. We assume that time is discrete so that the game players can take actions only in distinct time slots. Assume the time interval that a game player can play the game is $[1, T]$ with $T$ time slots. In a time slot $t$, the set of the strategies that a game player $i$ takes is $S(i) := \{NULL, SELL(j), BUY(j)\}$ where $j = 1, 2, \ldots, K$ is a gaming asset. $NULL$ indicates that the player $i$ does nothing. $SELL(j)$ indicates that the player $i$, who owns the game asset $j$, sells this asset to the platform $P$ at $j$'s price at that time. $BUY(j)$ indicates that the player $i$ buys the asset $j$ from the platform at the asset's price at that time. We remark here that in such a game, we relay the actions among the game players to the actions between the game players and the platform. For example, if a game player $i$ sells an asset $j$ to a game player $i'$, such a trade is described as the strategy $SELL(j)$ for the player $i$ followed by the strategy $BUY(j)$ for the player $i'$.

The game platform has the set of strategy $S(P) :=  \{NULL, PSELL(j, w(j, t)), PBUY(j)\}$, where $w(j, t) = 0, 1, \ldots, U$ is the amount of \emph{incentive NFTs} released in a time slot $t$ for selling the asset $j$. Here, $U = \max_j U(j)$. For any game player, if he/she buys the asset $j$ at time $t$, then the platform takes the action $PSELL(j, w(j, t))$ and what the game player pays to the platform $P$ is the difference of $j$'s current price at time $t$ minus $w(j, t)$. At a time $t$, we use $\alpha(j, t)$ to denote the number of game players who take the action $BUY(j)$. 
If nobody wants to have $j$ (i.e., $\alpha(j, t) = 0$), then the game player, who owns $j$, takes the action $SELL(j)$. At time $t$, if $\alpha(j, t) \ge 1$, then the platform, with a probability $\frac{1}{\alpha(j, t)}$, takes the action $PSELL(j, w(j, t))$ to sell the asset $j$ to a game player $i$ who takes the action $BUY(j)$. $PBUY(j)$ indicates that the platform buys the asset $j$ from a game player who owns $j$ at the asset $j$'s price at that time. $PBUY(j)$ happens when $\alpha(j, t) = 0$. We assume that the $K$ gaming assets are released into the game at the beginning of the time slot $1$ and these gaming assets are owned by the platform initially.

Thirdly, we describe the assets' values. For each asset $j$, it has a \emph{basic profile} describing its prices in the range $[1, T]$. This basic profile describes an asset's \emph{basic value} concerning all external factors and it is independent of the strategies taken within the game between the game players and the platform. Let $v(j, t)$ denote the basic value of the asset $j$ at time $t$, where $j = 1, \ldots, K$ and $t = 1, \ldots, T$. The basic profile can be predicted and thus, we have a mild assumption: \emph{All the game players and the platform have the basic profile information}. At a time $t$, if there are $\alpha(j, t)$ game players want to buy the asset $j$, then the asset $j$'s value is increased by $c \frac{\alpha(j, t)}{N}$ times of its basic value where $c > 0$ is a constant depending on the game. We use $v^*(j, t) = v(j, t) \left(1 + c \frac{\alpha(j, t)}{N}\right)$ to denote the \emph{traded value} of NFTs for the asset $j$. Note $v^*(j, t) = v(j, t)$ if $\alpha(j, t) = 0$ and $v^*(j, t) = (1 + c) v(j, t)$ if $\alpha(j, t) = N$. The traded values, rather than the basic values, are the assets' price for the game player at a time. Equivalently to say, if a player takes the action $BUY(j)$ successfully to get an asset $j$ at time $t$ and he/she sells via $SELL(j)$ this asset immediately to the platform, then he/she immediately gets a payoff $v^*(j, t) - v(j, t) = c \frac{\alpha(j, t)}{N}$. We can regard $c \frac{\alpha(j, t)}{N}$ as the incentives offered by the platform. (We can extend this trade value to $v^*(j, t) = v(j, t) \left(1 + c \frac{\alpha(j, t)}{N}\right)$ minus the total incentive NFTs.)

Fourthly, we examine the \emph{payoffs} for the game players and the platform. A game player \emph{objective} is to earn as many NFTs as possible through playing the game. Equivalent to say, a game player tries to maximize the total NFTs represented by the gaming assets he/she owns. The game designer's \emph{objective} is to maximize the total traded values of assets $\sum^K_{j = 1} \sum^T_{t = 1} v^*(j, t)$.
\vspace{-5pt}

\subsubsection{An incentive mechanism in a P2E game.}

In this section, we describe an incentive mechanism. Consider a time $t \in [1, T]$ and an asset $j$. The incentive mechanism consists of two parts:
\begin{enumerate}
    \item Assume $\alpha(j, t) = 0$: If the platform $P$ owns the asset $j$, then both the game players and $P$ take the strategy $NULL$. If the game player $i$ who owns the asset $j$, if any, sells the asset to the platform $P$ at the price $v(j, t)$. The platform $P$ takes the strategy $PBUY(j)$ at $v(j, t)$.

    \item Assume $\alpha(j, t) \ge 1$: The platform $P$, with a probability $\frac{1}{\alpha(j, t)}$ sells the asset $j$ to the game player $i$ who takes the action $BUY(j)$. The game player $i$ gets the asset with a value $v^*(j, t) = v(j, t) + v(j, t) \frac{c \alpha(j, t)}{N}$. In this situation, we regard that $v(j, t) \frac{c \alpha(j, t)}{N}$ as the \emph{incentive NFTs} issued by the platform for the asset $j$ given to the game player $i$. If the game player $i'$ ($\neq i$) owns the asset $j$, if any, then $i'$ must sell the asset to the platform $P$ at the price $v(j, t)$. The platform $P$ takes the action $PBUY(j)$ to buy the asset $j$ at the price $v(j, t)$.
\end{enumerate}

In the following, we prove that an equilibrium exists for this game. Before we prove the result, we show that there exists a plan for a game player to maximize his/her payoff, assuming there are no other game players. 

\begin{theorem}
Assume there is only one player and he/she has the basic profile for all the assets $j$ in the interval $[1, T]$, where $j = 1, \ldots, K$. Then a plan exists for a game player to maximize his/her total payoff in this interval.
\label{thm:1}
\end{theorem}

\begin{proof}
Consider a game player and an asset $j$. Let $OPT(j, t)$ denote the profit that the player makes at the end of the time slot $t$ by selling the asset $j$ to the platform in the time slot $t$. In this time slot $t$, $\alpha(j, t) = 0$. The time slots are partitioned into two types: In one type of time slots $T_1$, we have $\alpha(j, t_1) = 1$ where $t_1 \in T_1$ and in the other type of time slots $T_2 = [1, \ldots, T] \setminus T_1$, we have $\alpha(j, t_2) = 0$ where $t_2 \in T_2$. We thus have a recurrence
{\setlength{\abovedisplayskip}{0pt}
\setlength{\belowdisplayskip}{0pt}
\begin{align}
OPT(j, 0) &= 0,\notag\\
OPT(j, t) &= \max_{1 \le t' < t} \Bigl[OPT(j, t') \nonumber \\
&\quad + \Bigl(v(j, t) - \min_{t' \le t'' < t} v(j, t'')\Bigr)\Bigr],\quad t \ge 1.\notag
\end{align}}

The optimal solution is $\max_{t \in [1, T]} OPT(j, t)$. The algorithm finding the best trading time slots to maximize the total payoff of a game player is described in Algorithm~\ref{alg:points}. From this algorithm, we get the best trading time slots and the payoff that can be obtained at these trading points.

\begin{algorithm}[t]
\caption{Identify trading points for an asset $j$}
\label{alg:points}

\For{$j = 1$ \KwTo $K$}{
    
    Define $OPT(j, t)$ denote the profit that the player $i$ makes at the end of the time slot $t$ by selling the asset $j$ in the time slot $t$\;

    $X(j) \leftarrow \emptyset$\;
    $Y(j) \leftarrow \emptyset$\;
    $OPT(j,0) \leftarrow 0$\;

    \For{$t = 1$ \KwTo $T$}{
        
        Compute $OPT(j,t)$ according to Eq.~(\ref{eq:points})\;

        \If{$v(j,t) - \min_{t' \le t'' < t} v(j,t'') > 0$}{
            Find $t^*$ that achieves the minimum\;
            $X(j) \leftarrow X(j) \cup \{t^*\}$\;
            $Y(j) \leftarrow Y(j) \cup \{t\}$\;
        }
    }
}

Return $X(j)$, $Y(j)$, and $OPT(j,y(j))$ where $y(j)\in Y(j)$\;

\end{algorithm}

\begin{equation}
OPT(j,t)=
\max_{1 \le t' < t}
\left[
OPT(j,t')+
\left(
v(j,t)-
\min_{t' \le t'' < t} v(j,t'')
\right)
\right]
\label{eq:points}
\end{equation}

This dynamic programming algorithm has its running time $O(T^2 K)$.
\end{proof}
In the following, we calculate the incentive NFTs for the platform to issue so that the platform has the plan of maximizing his/her payoff. Note that in the single game-player setting, the traded value is $v^*(j, t) = (1 + c) v(j, t)$ (since $\alpha(j, t) = N = 1$) if the player takes the action $BUY(j)$  and is $v^*(j, t) = v(j, t)$ (since $\alpha(j, t) = 0, N = 1$) if the player takes the action $SELL(j)$ or $NULL$, for any time $t$ and any asset $j$.

From Algorithm~\ref{alg:points}, we know that buying time slots and selling time slots appear alternatively. There is nothing that a platform can do to incentivize a game player during the period from the last selling point in $Y(j)$ to the end of this game interval $T$. However, if we make Equation~(\ref{eq:points}) hold for every time slot $t''$ which is between a buying time slot in $X(j)$ and its following selling time slot in $Y(j)$, then the player is incentivized to buy the asset $j$ in all time slots in this period. Thus, we have the following result.

\begin{theorem}
Consider a time slot $t$. For a gaming asset $j$, the game designer takes the action $PSELL(j, w(j, t''))$ in a time slot $t'' \notin \left(X(j) \cup Y(j)\right)$ so that for any neighboring time slots $t' \in X(j)$ and $t \in Y(j)$ ($t'$ and $t$ are two neighboring time slots in $X(j) \cup Y(j)$), we have
% \vspace{-20pt}
\footnotesize{\begin{align}
OPT(j, t') + \left(v(j, t) - v(j, t'') + w(j, t'')\right) = OPT(j, t) & & 1 \le t' < t 
\end{align}}
\label{thm:platform}
\end{theorem}
Theorem ~\ref{thm:platform} gives the platform a way of incentivizing the game players to take the action $BUY(j)$. Let us consider a multiple-player scenario. We argue in the following that the $PSELL(j, w(j, t))$ policy (with a probability $\frac{1}{\alpha(j, t)}$, taking the action $PSELL(j, w(j, t))$ to sell the asset $j$ to the game player $i$ who takes the action $BUY(j)$) ensures that all players are incentivized and they reach an equilibrium when $w(j, t) = c \frac{\alpha(j, t)}{N}$.

\begin{theorem}
Assume there are multiple, say $N$, players. There exists an equilibrium for all the game players to take the strategy $BUY(j)$ in all the time slots $X(j)$, $Y(j)$, and those time slots in which the platform releases incentive tokens.
\label{thm:3}
\end{theorem}

\begin{proof}
Consider the game players and an asset $j$. If Theorem~\ref{thm:3} holds for one asset $j$, then it holds for all assets. Consider all the time slots $Q(j)$ including $X(j)$, $Y(j)$, as well as those time slots that the platform issues incentive NFTs. To prove Theorem~\ref{thm:3}, we only need to show in the following that any player does not want to take the action $NULL$ in one of the time slots in $Q(j)$.

We use indicator variables to calculate a game player's expected payoff in the time slots $Q(j)$. Define
\[
Z(i,t,t') =
\begin{cases}
1, & 
\begin{aligned}
& \text{if player $i$ takes action $BUY(j)$ for some $j$} \\
& \text{at time slot $t \in Q(j)\setminus Y(j)$ and takes action} \\
& \text{$SELL(j)$ at time slot $t'$ with $t' > t$}
\end{aligned}
\\[6pt]
0, &
\begin{aligned}
& \text{otherwise (player $i$ takes $NULL(j)$} \\
& \text{at time slots $t$ and $t'$)}
\end{aligned}
\end{cases}
\]

Let $W(i, j)$ denote the expected payoff for the game player $i$ on playing the gaming asset $j$. Then we have
\begin{align}
W(i,j) = & \sum_{\substack{t \in Q(j)\\ t' \in Y(j)\\ t' > t}}
\mathbb{E}\!\left[ Z(i,t,t') \left( v(j,t') - v(j,t) + c\frac{\alpha(j,t)}{N} \right) \right] \notag\\
= & \sum_{\substack{t \in Q(j)\\ t' \in Y(j)\\ t' > t}}
\Biggl( \frac{1}{\alpha(j,t)}\bigl( v(j,t') - v(j,t) \bigr) \notag\\
  & \qquad + \frac{c}{N}\,\mathbb{E}\!\left[ Z(i,t,t') \,\alpha(j,t) \right] \Biggr) \notag\\
= &
\begin{cases}
\displaystyle
\sum_{\substack{t \in Q(j)\\ t' \in Y(j)\\ t' > t}}
\left( \dfrac{1}{\alpha(j,t)}\bigl( v(j,t') - v(j,t) \bigr) + \dfrac{c}{N} \right),
& \\\text{if $i$ takes $BUY(j)$ at time $t$}, \\[6pt]
0, & \\\text{if $i$ takes $NULL$ at time $t$}.
\end{cases}
\label{eq:expected}
\end{align}
Equation~(\ref{eq:expected}) implies that to maximize a game player's payoff, he/she should not take the $NULL$ action in those $Q(j)$-time slots $t$, given $v(j, t') < v(j, t)$. Therefore, an equilibrium exists for the game players only. Also, the platform has the objective of
\[
\max \sum_{t \in Q(j)} v^*(j, t)
= \max \sum_{t \in Q(j)} \left(v(j, t) + c \frac{\alpha(j, t)}{N}\right)
\]
for the gaming asset $j$. This objective is maximized when $\alpha(j, t)$ is maximized. Thus, the platform and all the game players reach an equilibrium.
\end{proof}
\subsection{Simulation results}\label{sec:sml}

Based on the incentive model above, we develop a simulator in Matlab based on the data in Table~\ref{tab:profits},  and calculate the constant $c$ when no incentive NFTs exists in the games (\ie $w(j,t) = 0$ where $j = 1,...,K$ and $t = 1,...,T$). We exclude the games with existing incentive mechanism when trading the NFTs, \eg SOR rewards players whose monthly trading volume has reached a threshold with utility tokens, thus, we exclude it from the simulation. Since the NFTs are sold from the game platform to the players, we set the initial basic profile as the average selling profits,\ie the basic profile when $t = 1$. When $t = 2,...,T$, the basic profile of the assets is evenly changed based on the difference between the average selling profits and the average trade price. In other words, when t = T, the basic profile of the assets is the average trade price. Then we run our model with the same set of the basic profile, the calculated $c$, and the incentive, where the incentive amount is calculated based on  Theorem~\ref{thm:platform}. We set N = 100, K = 100, L =0, U = 10000, and T = 12 as the parameter when we run our model.

\begin{table*}
\centering
\footnotesize
\caption{Model simulation result with and without incentives}
\label{tbl:simulation}
\begin{tabular}{c|c|c|c||c||c||c||c}
Game&Sell profit&Trade Price&Profits&\makecell{Simulated \\Profits (without w)}&	c &\makecell{Simulated \\Profits (with w))}&Improvement (\%)\\\hline
BLNKS&	282.81&	361&	19.22&	19.25&	0.08&	22&	14.29\\
CK&	22.93&	35.57&	78.64&8.91&	0.95&	115.22&	46.01\\
LOKL&	933.46&	1315.28&	721.38&	719.21&	0.67&	2870.81&	299.16\\
LOKR&	13.54&	32.01&	82.18&	81.6&	1.1&	195.96&	140.15\\
TANKS&	465.82&	473.61&	207.07&	206.11&	0.61&	372.28&	80.62\\
LANDS&	6732.83&	6148.78&	156.24&	156.17&	0.04&	210.21&	34.60\\
AXIE&	331.65&	203.43&	196.93&	193.65&	0.98&	391.55&	102.19\\

\end{tabular}
\end{table*}

Table~\ref{tbl:simulation} shows the simulation results. We can see that when no incentive is used, the simulated average profits are within 2\% of the real-world profits, indicating our model can precisely abstract the trading activities in the NFT games. Furthermore, the simulated profits when the incentive is employed show that our incentive mechanism can effectively improve the profits that players can earn, which suggests the incentive amount based on our model can successfully encourage trading activities and increase the rarity of the assets (\ie $\alpha$).
\section{Discussion}
\label{sec:discussion}

\noindent\textbf{Limitation of our study.} With a lot of NFT games in the wild, we only scratched the surface and focus on  NFT games with full transaction history from public sources, e.g., the Ethereum blockchain. 
Furthermore, while some data is available for a game, it may not be available for other games. 
As a result, we have to exclude them during the study.
Also, the data we have collected from public sources often are of coarse granularity without detailed breakdowns. 
Such obstacles have hindered us to paint a global picture of NFT games. As such, while we analyze the data, we have also mainly focused on the these popular games that we have more and complete information. 
As it is still in the early stage of NFT games, we hope our study can still provide some insights regarding the play-to-earn model and areas to improve in the future.

\noindent\textbf{What can be improved for P2E?} Blockchain support enables the P2E, which is a great incentive to players. However, based on the preliminary analysis we have conducted in this study, P2E in reality is not (yet) working as it has been designed for or is not achieving its original goal.
While this may be due to the premature nature of these games, however, if this trend continues in the future, it may lose attraction to players. A major observation we have found from the analysis is that there are always some dominating players owning a significant portion of the NFTs in the game while only a small percentage of NFTs are distributed among other ordinary players.
This discourages the participation of ordinary players and is in the opposite direction of Internet democracy which can be naturally supported by the blockchain. 
From the profitability model and the TLA+ model, we can conclude that the estimated players' profit is related to the NFT circulation rate and the average NFT price. Incentive mechanisms would make a positive impact on the sustainability of NFT games, but cannot guarantee an increase of profits for the ordinary players. We would like to further explore how much the external factors (\ie big events in the big days) can influence the NFT price in the future, which completes the final piece of the puzzle regarding the players' profits.
\section{Related Work}
\label{sec:related}

The increasing popularity of NFT games has attracted some prior studies. 
For example, Vidal-Tomás~\cite{5vidal2022new} analyzed the play-to-earn NFT games via social media data as trends on social media can reveal the attention the game received. Boonparn et al.~\cite{boonparn2022social} developed a web application that can retrieve data from social media to analyze NFT games. 
Similarly, Francisco \etal~\cite{7francisco2022perception} performed a qualitative analysis based on the interviews they conducted with the Axie Infinity players. 
Francisco \etal~\cite{7francisco2022perception} confirmed the risks on financial instability from the Axie Infinity players.

Alam~\cite{alam2022understanding} also studied Axie Infinity by fetching the data from the blockchain it resides (\ie Ronin~\cite{ronin}%\rz{reference?}
) and analyzing the transaction history,  and found that the top 1\% of the wallets earned 60.1\% of total profits in the marketplace of the exchange between SLP and WETH. Alam~\cite{alam2022understanding} believed this is caused by the high payoff rate at the early stage of Axie Infinity and suggested the concerns of a possible Ponzi scheme. 
Delic \& Delfabbro~\cite{delic2022profiling} conducted a qualitative analysis with the players of Axie Infinity and suggested that the close monitoring of NFT games is necessary, particularly regarding the potential exploitation of labor, as players might be enticed to prioritize grinding for rewards over playing for pure enjoyment.
Seifoddini~\cite{seifoddini2022multi} developed a framework for rating the NFT games using multiple criteria (\eg the token market cap,  the security score published by a security monitor company of DeFi projects, and whether the NFT game has a frequently updated white paper/green paper). 
However, some of the criteria (\eg the number of Google Scholar citations of the game) cannot be used as a piece of strong evidence showing the success of an NFT game.
Recently, Yang \& Wang~\cite{yang2023non} performed a literature review of 23 NFT game articles and found that social media frequently appeared in these articles and that the majority of the articles proposed some initial ideas, and more consideration is needed, \eg the criticisms of Ponzi schemes.  
Prior researchers~\cite{ecr21scholten2019ethereum,ecr22lee2019blockchain,ecr24fritsch2024analyzing,ecr25mohima2023exploring} noticed the low earning rate in some NFT games and even compared it with gambling and speculation, however, these works focused on one game only and can hardly conclude the behaviors of the industry. We have the same opinion on these works about the current P2E from a comprehensive dataset from 12 NFT games, but the opposite perspective on the future of P2E. We have proved that the failure of the current P2E can be reverted by an optimized incentive mechanism, which is the focus of our study.

\section{Conclusion}
\label{sec:conclusion}

Along with cryptocurrencies, NFTs and NFT games have emerged in recent years.
NFT games, supported by blockchain techniques, have enabled a new model, play-to-earn, which sharply differs from traditional games where users often need to pay-to-play. While numerous users have been attracted to various NFT games, little is known about its underpinnings. In this study, we have focused on understanding the effectiveness of the play-to-earn model in practice. 
Our empirical study leads to several findings that could provide valuable insights into the development of NFT games in the future. 
We found that current incentives provide players merely earnings thus should be enhanced. Motivated by the findings, we have explored the 
incentive mechanism of NFT games and proposed an incentive mechanism model to evaluate the impact of the incentive mechansim.  Our model results show that proper incentives are notably effective in improving the profits that players earn from these NFT games. 

\bibliographystyle{IEEEtran}
\bibliography{sample}

@String{BIT = "{BIT}" }

@String{Computing = "Computing" }

@String{Computer = "{IEEE} Computer" }

@String{Springer = "Springer-Verlag" }

@article{4delfabbro2022understanding,
  title={Understanding the mechanics and consumer risks associated with play-to-earn (P2E) gaming},
  author={Delfabbro, Paul and Delic, Amelia and King, Daniel L},
  journal={Journal of Behavioral Addictions},
  volume={11},
  number={3},
  pages={716--726},
  year={2022},
  publisher={Akad{\'e}miai Kiad{\'o} Budapest}
}

@article{5vidal2022new,
  title={The new crypto niche: NFTs, play-to-earn, and metaverse tokens},
  author={Vidal-Tom{\'a}s, David},
  journal={Finance research letters},
  volume={47},
  pages={102742},
  year={2022},
  publisher={Elsevier}
}

@inproceedings{6boonparn2022social,
  title={Social data analysis on play-to-earn non-fungible tokens (NFT) games},
  author={Boonparn, Pahukun and Bumrungsook, Phuriput and Sookhnaphibarn, Kingkarn and Choensawat, Worawat},
  booktitle={2022 IEEE 4th Global Conference on Life Sciences and Technologies (LifeTech)},
  pages={263--264},
  year={2022},
  organization={IEEE}
}

@article{nadini2021mapping,
  title={Mapping the NFT revolution: market trends, trade networks, and visual features},
  author={Nadini, Matthieu and Alessandretti, Laura and Di Giacinto, Flavio and Martino, Mauro and Aiello, Luca Maria and Baronchelli, Andrea},
  journal={Scientific reports},
  volume={11},
  number={1},
  pages={20902},
  year={2021},
  publisher={Nature Publishing Group UK London}
}

@article{7francisco2022perception,
  title={The perception of Filipinos on the advent of cryptocurrency and non-fungible token (NFT) games},
  author={Francisco, Ryan and Rodelas, Nelson and Ubaldo, John Edison},
  journal={arXiv preprint arXiv:2202.07467},
  year={2022}
}

@article{alam2022understanding,
  title={Understanding the economies of blockchain games: An empirical analysis of Axie Infinity},
  author={Alam, Omar},
  journal={Distributed Computing Group Computer Engineering and Networks Laboratory ETH Z{\"u}rich.—2022.—URL: https://pub. tik. ee. ethz. ch/students/2},
  year={2022}
}

@INPROCEEDINGS{430-1,
  author={Karapapas, Christos and Syros, Georgios and Pittaras, Iakovos and Polyzos, George C.},
  booktitle={2022 4th Conference on Blockchain Research \& Applications for Innovative Networks and Services (BRAINS)}, 
  title={Decentralized NFT-based Evolvable Games}, 
  year={2022},
  volume={},
  number={},
  pages={67-74},
  doi={10.1109/BRAINS55737.2022.9909178}
  }

@misc{430-3, 
title={NFT Game Stats 2023}, 
url={https://bithotel.io/blog/nft-game-stats-2023/\#:~:text=So the market will double in size through, games in development at the time of writing.}, 
journal={Bit Hotel}, 
author={MarketingDep}, 
year={2023}, 
month={Jan}
}

@misc{Axie,
  author = {Sky Mavis},
  title = {Axie Infinity - Battle, Collect, and Trade Collectible NFT Creatures},
  note = {https://axieinfinity.com, Accessed Date: 2023-05-19},
  year = {2023}
}

@misc{benjibananas,
  title = {Benji Bananas},
  note = "\url{https://benjibananas.com/}, Date Accessed: 2023-05-19",
  author = {Benji Bananas Ltd},
  year = {2022}
}

@misc{blankos,
  title = {Blankos},
  note = "\url{https://blankos.com/, Accessed  Date: 2023-05-19}",
  author = {Blankos},
year = {2022}
}

@misc{cryptokitties,
  title = {CryptoKitties | Collect and breed digital cats!},
  note = {https://www.cryptokitties.co/, Accessed Date: 2023-05-19},
  year = {2017},
  author = {CryptoKitties}
}

@misc{animocabrands,
  title = {Animoca Brands},
  note ={https://www.animocabrands.com, Accessed Date: 2023-07-29},
  year = {2023},
  author = {Animoca Brands}
}

@misc{decentraland,
  title = {Welcome to Decentraland},
  note = {https://decentraland.org/, Accessed Date: 2023-05-19},
 year = {2023},
author = {Decentraland}
}

@misc{ethermon,
  title = {Ethermon.io - Decentralized World of Ether Monsters},
  note ={https://ethermon.io/, Accessed Date: 2023-05-19},
  year = {2022},
author = {Ethermon}
}

@misc{F1DeltaTime,
  author = {STEVE MCCASKILL},
  title = {F1 Delta Time blockchain game shuts down},
  note = {https://www.sportspromedia.com/news/f1-delta-time-blockchain-game-shuts-down/, Accessed Date: 2023-04-05},
  year = {2022}
}

@misc{leagueofkingdoms,
  title = {League of Kingdoms | MMO Strategy Game on Blockchain},
  note  ={https://www.leagueofkingdoms.com/, Accessed Date: 20223-05-19},
  year = {2022},
author = {League of Kingdoms}
}

@misc{mycryptoheroes,
  title = {My Crypto Heroes (MCH) | Crypto game from Japan!},
  note = {https://www.mycryptoheroes.net/, Accessed Date: 2023-05-19},
  year = {2023},
author = {MCH Corporation}
}

@misc{sorare,
  title = {www.sorare.com - Official Site},
  note = {https://sorare.com/, Accessed Date: 2023-05-19},
  year = {2023},
author = {Sorare}
}

@misc{spidertanks,
  title = {Spider Tanks | Welcome to the Brawl},
  note = {https://www.spidertanks.game/, Accessed Date: 2023-05-19},
  author = {Spider Tanks},
year ={2022}
}

@misc{sandbox,
  title = {The Sandbox Game — User-Generated Crypto and Blockchain ...},
  note ={https://www.sandbox.game/en/, Accessed Date: 2023-05-19},
  author = {Sandbox},
year ={2018}
}

@article{bitcoin2008,
  title={Bitcoin: A peer-to-peer electronic cash system},
  author={Nakamoto, Satoshi},
  journal={Decentralized business review},
  pages={21260},
  year={2008}
}

@article{ethereumwhitepaper,
  title={A next-generation smart contract and decentralized application platform},
  author={Buterin, Vitalik and others},
  journal={white paper},
  volume={3},
  number={37},
  pages={2--1},
  year={2014}
}

@misc{ethereum,
  title = {ethereum.org},
  note = {https://ethereum.org/en/, Accessed Date: 2023-05-19},
  author = {Ethereum},
year = {2014}
}

@misc{ronin,
  title = {Ronin},
  note ={https://roninchain.com, Accessed Date: 2023-05-19},
  author = {Sky Mavis},
year = {2022}
}

@misc{CommunityGrants,
  title = {Community Grants},
  note ={https://docs.decentraland.org/player/general/dao/grants-v1/community-grants/, Accessed Date: 2023-05-21},
  year = {2022},
author = {Decentraland}
}

@misc{TheSandboxAmbassadorProgram,
  title = {The Sandbox Ambassador Program},
  note = {https://medium.com/sandbox-game/the-sandbox-ambassador-program-4f00bb5c9c24, Accessed Date: 2023-05-21},
  year = {2023},
author = {the Sandbox}
}

@misc{decentralandwork,
  title = {This Casino in Decentraland Is Hiring (for Real)},
  author = {Danny Nelson},
  year ={2021},
  note = {https://www.coindesk.com/tech/2021/03/18/this-casino-in-decentraland-is-hiring-for-real/ Accessed Date: 2023-05-21}
}

@misc{erc20,
  title = {ERC-20 TOKEN STANDARD},
  note = {https://ethereum.org/en/developers/docs/standards/tokens/erc-20/, Accessed Date: 2023-05-22},
  author = {Ethereum},
year = {2018}
}

@misc{erc721,
  title = {ERC-721 NON-FUNGIBLE TOKEN STANDARD},
  note ={https://ethereum.org/en/developers/docs/standards/tokens/erc-721/, Accessed Date: 2023-05-22},
    author = {{Ethereum}},
year = {2018}
}

@misc{erc1155,
  title = {ERC-1155 MULTI-TOKEN STANDARD},
  note ={https://ethereum.org/en/developers/docs/standards/tokens/erc-1155/, Accessed Date: 2023-05-22},
  author = {Ethereum},
  year = {2018}
}

@online{roninAPI,
  title = {Ronin Rest API},
   howpublished = "\url{https://documenter.getpostman.com/view/23367645/2s7YfPfEGh\#3d1a1e96-4348-4340-8235-fa9f89eed422, Accessed Date: 2023-05-22}",
  year = {2023},
author = {{Sky Mavis}}
}

@misc{yahoo,
  title = {Yahoo Finance - Stock Market Live, Quotes, Business and Finance News},
  note = {https://finance.yahoo.com, Accessed Date: 2023-05-22},
  year = {2023},
author = {Yahoo Finance}
}

@misc{etherscan,
  title = {Ethereum (ETH) Blockchain Explorer},
  note = {https://etherscan.io, Accessed Date: 2023-05-22},
  author = {Etherscan},
year = {2018}
}

@misc{coinbase,
  title = {Coinbase - Buy and Sell Bitcoin, Ethereum, and more with trust},
  note ={https://www.coinbase.com/, Accessed Date: 2023-05-22},
  author = {Coinbase},
year = {2023}
}

@article{dao,
  title={Decentralized autonomous organizations: Concept, model, and applications},
  author={Wang, Shuai and Ding, Wenwen and Li, Juanjuan and Yuan, Yong and Ouyang, Liwei and Wang, Fei-Yue},
  journal={IEEE Transactions on Computational Social Systems},
  volume={6},
  number={5},
  pages={870--878},
  year={2019},
  publisher={IEEE}
}

@inproceedings{yang2023non,
  title={Non-Fungible Token (NFT) Games: A Literature Review},
  author={Yang, Yann-Jy and Wang, Jing-Lun},
  booktitle={2023 International Conference On Cyber Management And Engineering (CyMaEn)},
  pages={251--254},
  year={2023},
  organization={IEEE}
}

@inproceedings{boonparn2022social,
  title={Social data analysis on play-to-earn non-fungible tokens (NFT) games},
  author={Boonparn, Pahukun and Bumrungsook, Phuriput and Sookhnaphibarn, Kingkarn and Choensawat, Worawat},
  booktitle={2022 IEEE 4th Global Conference on Life Sciences and Technologies (LifeTech)},
  pages={263--264},
  year={2022},
  organization={IEEE}
}

@article{delic2022profiling,
  title={Profiling the Potential Risks and Benefits of Emerging “Play to Earn” Games: a Qualitative Analysis of Players’ Experiences with Axie Infinity},
  author={Delic, Amelia J and Delfabbro, Paul H},
  journal={International Journal of Mental Health and Addiction},
  pages={1--14},
  year={2022},
  publisher={Springer}
}

@article{seifoddini2022multi,
  title={A multi-criteria approach to rating Metaverse games},
  author={SE{\.I}FODD{\.I}N{\.I}, Jalal},
  journal={Journal of Metaverse},
  volume={2},
  number={2},
  pages={42--55},
  year={2022}
}

@online{bbpE,
  title={How much can you earn from Blankos Block Party Party Pass? Party Pass Plus worth it? playblankos},
  author = {Akadonniedanko},
  url = {https://www.youtube.com/watch?v=Vr85uJAqwjQ},
  note = {Accessed Date: 2023-05-19},
  year = {2022}
}

@mis{benE,
  title={Benji Bananas Announces Grand Opening of Benji Shop},
  author = {Akadonniedanko},
  note = {https://playtoearn.net/news/benji-bananas-announces-grand-opening-of-benji-shop, Accessed Date: 2024-01-30},
  year = {2022}
}

@online{emonE,
  title={Learn How to Play Ethermon Play to Earn Game | Spintop},
  url = {https://spintop.network/gamepedia/games/ethermon\#},
  note = {Accessed Date: 2024-01-30},
author = {spintop},
year = {2024}
}

@online{dlandE,
  title={How to make money with Decentraland},
  url = {https://medium.com/@innovativeartist/how-to-make-money-with-decentraland-83c394368b5c},
  author = {David Schuyler Franco},
  urldate = {2023-04-23},
year = {2023}
}

@online{lokE,
  title={Earn by Playing League of Kingdoms in 3 Simple Ways},
  url = {https://p2enews.com/earn-by-playing-league-of-kingdoms-in-3-simple-ways/\#:~:text=While owning land is a, secure some profit while trading!},
  author = {Jaylance Lorete},
  note = {Accessed Date: 2024-01-30},
year = {2024}
}

@online{report,
  title={Non-Fungible Token (NFT) Market Size | 39.4\% CAGR},
  url = {https://market.us/report/non-fungible-token-nft-market/},
  urldate = {2025-01-31}
}

@article{borri2022economics,
  title={The economics of non-fungible tokens},
  author={Borri, Nicola and Liu, Yukun and Tsyvinski, Aleh},
  journal={Available at SSRN},
  volume={4052045},
  year={2022}
}

@inproceedings{ecr21scholten2019ethereum,
  title={Ethereum crypto-games: Mechanics, prevalence, and gambling similarities},
  author={Scholten, Oliver James and Hughes, Nathan Gerard Jayy and Deterding, Sebastian and Drachen, Anders and Walker, James Alfred and Zendle, David},
  booktitle={Proceedings of the annual symposium on computer-human interaction in play},
  pages={379--389},
  year={2019}
}

@inproceedings{ecr22lee2019blockchain,
  title={Is a blockchain-based game a game for fun, or is it a tool for speculation? An empirical analysis of player behavior in crypokitties},
  author={Lee, Jaehwan and Yoo, Byungjoon and Jang, Moonkyoung},
  booktitle={The Ecosystem of e-Business: Technologies, Stakeholders, and Connections: 17th Workshop on e-Business, WeB 2018, Santa Clara, CA, USA, December 12, 2018, Revised Selected Papers 17},
  pages={141--148},
  year={2019},
  organization={Springer}
}

@article{ecr24fritsch2024analyzing,
  title={Analyzing voting power in decentralized governance: Who controls DAOs?},
  author={Fritsch, Robin and M{\"u}ller, Marino and Wattenhofer, Roger},
  journal={Blockchain: Research and Applications},
  volume={5},
  number={3},
  pages={100208},
  year={2024},
  publisher={Elsevier}
}

@phdthesis{ecr25mohima2023exploring,
  title={Exploring attacks in the NFT gaming industry: A study of risks and mitigation strategies},
  author={Mohima, Zarin Rahman and Bin Bashar, Syed Ziaul and Muktadir, Rawnak and Hossain, Amirah and Mahmud, Shafin},
  year={2023},
  school={Brac University}
}

\appendices
\section{Background on NFT Games}
\label{sec:background-appendix}
In this section, we provide a brief background about the NFT gaming platforms. In particular, we discuss the incentive models used in NFT gaming platforms. 

% \subsection{Economical Models and Play-to-earn}
Having looked at the token types in NFT games, we now describe the economic model that underpins NFT games. In particular, we compare the NFT game incentive model known as P2E with traditional game models (\ie pay-to-play and free-to-play).

As aforementioned, NFT games are also called play-to-earn games, enabled with blockchain support.
NFT games adapt to the ecological environment of the blockchain  because
(1) the game and its in-game props, assets, and other natural virtual attributes make the blockchain a perfect platform for trading and transactions;
and (2) the play-to-earn business model can be easily implemented based on the blockchain ecosystem.

Gaming platforms, as a type of entertainment program, are essentially virtual assets, and so are the players' game accounts, props, or properties in the game account. However, the trading of assets within the gaming platforms is largely restricted by game companies. For better profits, not all gaming platforms will open trading channels for their players. Under such a circumstance, if players want to obtain a specific game item or property, and there is no official way to achieve that, players will consider trading the game accounts. This seriously damages the interests of game companies. At the same time, private transactions between players cannot be guaranteed, and fraud may occur. However, if gaming platforms introduce the blockchain when important game props are tokenized as NFTs, there are several advantages:
1. Game props, as a virtual asset, can naturally be transferred into NFT; 2. The transaction of props can be circulated peer-to-peer in a trustless environment without time, place, and identity restrictions; 3. These game props' transactions are open and cannot be tampered with. Game developers do not need to maintain public transaction records and backups. Therefore, it is a reasonable attempt to add blockchain to the architecture of video games.

On the other hand, the business models of traditional games are mainly divided into pay-to-play and free-to-play. Pay-to-play requires players to pay money to obtain the qualifications to participate in the game or obtain the props in the game while free-to-play allows players to play games for free. However, free-to-play requires players to make efforts, such as participating in game activities (paying time and effort)
to get props or a more satisfying gaming experience. These two business models naturally oppose players and game companies into the roles of consumers and service providers, and the benefits of the two roles may conflict.
However, when the blockchain is embedded into the architecture of video games, players can freely trade NFTs. After the NFTs are sold to players, their attributes will not change, and they will no longer be controlled and maintained by game companies. On the one hand, any limited-version game props have a potential appreciation, which requires the game company to limit the quantity of the NFT released. On the other hand, once the NFT are sold to players, the game company loses the monopoly sales channel for game props right away that requires a moderate selling price of the NFT. Under such a circumstance, attracting more players into the community benefits the game company and the blockchain ecosystem. 
Therefore, the play-to-earn (P2E) business model was proposed to attract players to join.

\section{Description of NFT Games}
\label{sec:gameintro}
\textbf{Axie Infinity}~\cite{Axie} is a single-player Player vs. Player (PvP) game developed by Sky Mavis. The ERC-721 compatible NFTs named \texttt{Axies} are the Pokemon-like digital pets that players can control to battle with others. There are 2 ERC-20 compatible tokens in Axie Infinity. \texttt{AXS} is the govenance token and \texttt{SLP} is the utility token. Although it is not on Ethereum now, it used to be on Ethereum and now is on Ethereum-linked sidechain, Ronin.
\textbf{Benji Bananas}~\cite{benjibananas} is a mobile action game. In the game, the players control the action of Benji the monkey and his friends to let them jump from vine to vine through the jungle. The ERC-1155 NFT in Benji Bananas, named \texttt{Membership Pass}, represents qualification to participate in the Play-to-Earn ecosystem and receive any earning tokens. \texttt{BENJI} is the ERC-20 token and players would receive \texttt{BENJI} through play. \texttt{BENJI} tokens are exchangeable to other tokens developed by the same company, Animoca Brands\cite{animocabrands}.

\textbf{Blankos Block Party}~\cite{blankos} is a PC open-world multiplayer game where players can design their own digital vinyl toys and use it as the avatar of the owner when racing, shooting, and actions in the game world. The digital vinyl toys are named \texttt{Blankos}.

\textbf{CryptoKitties}~\cite{cryptokitties}, self-claimed as "the first blockchain game in the world"\cite{cryptokitties}, is a pet breed game where the players can collect and breed digital kitties. \texttt{Kitties} are ERC-721 NFTs which players can trade, breed, and interact (\eg, feed).

\textbf{Decentraland}~\cite{decentraland} is a 3D metaverse game founded in 2015 and launched on Ethereum in 2017. Every chunk of the virtual land in the metaverse is either a public place or land that can be sold as private land. The players can purchase the virtual land and build on their land. Players can also interact with each other in public places or on private land. The NFT in Decentraland, named \texttt{LAND}, represents the ownership of the virtual land. \texttt{MANA} is the ERC-20 token, the virtual currency to purchase \texttt{LAND} in Decentraland. The owner of either token is granted voting power based on the amount and the weight of each type of token.

\textbf{Ethermon}~\cite{ethermon}, originally named as "Etheremon" in 2017, which self-claimed as "the first PVP battle game on Ethereum"~\cite{ethermon}. It has the ERC-721 NFTs \texttt{Monsters} and ERC-20 token \texttt{EMON}. Players trade and control their \texttt{Monsters} to battle with \texttt{Monsters} of other players to earn \texttt{EMON} and other prizes. In 2019, the original team left and the top players took over the game and changed the name as "Ethermon".

\textbf{F1 Delta Time}~\cite{F1DeltaTime} was launched on 2019 and is reported as the first licensed NFT game. Players use Ethers to purchase the \texttt{REVV} utility token to trade the NFTs such as cars and drivers. F1 Delta Time has shut down in March, 2022 due to the loss of the license. 

\textbf{League of Kingdoms}~\cite{leagueofkingdoms} is a  massively multiplayer online (MMO)  strategy game. Players build their kingdom and fight with other kingdoms for dominion. There are multiple ERC-721 NFTs, \texttt{LAND} represents the virtual land in the game. Players build their kingdoms on their \texttt{LANDs} and the higher level of the \texttt{LANDs} is, the more profits players can earn. \texttt{Resource} is another collection of ERC-721 NFTs  representing four types of in-game resources: Food, Stone, Lumber, and Gold, respectively. \texttt{LOKA}, the ERC-20 utility token of League of Kingdoms, can be used to trade NFTs, representing the voting power inside the game, and as the one of the rewards that NFT owners may earn.

\textbf{My Crypto Heroes}~\cite{mycryptoheroes}, self-claimed as the blockchain game with the most daily active users as of 08/21/2019, is a role-playing game (RPG) where players can control the historically inspired heroes to complete quests and battle against other players.  \texttt{Hero} is the ERC-721 NFT that represents the historical inspired heroes that players control; the equipments on a hero belong to the \texttt{Extension} collection of ERC-721 NFT. \texttt{MCHC} is the governance token and \texttt{RAYS} is the utility token.

\textbf{Sorare}~\cite{sorare} is a sports card play game. Digital cards that represent the athletes are the ERC-721 NFTs in the \texttt{SOR} collection. Players trade the cards and form their own sports team to compete with others. Sorare has the partnership with over 300 sports clubs and leagues such as NBA and MLB.

\textbf{Spider Tanks}~\cite{spidertanks} is a PVP Brawler PC game developed by GAMEDIA. \texttt{Tanks} are ERC-1155 NFTs that can be traded, rented, and upgraded by the players and used for competition with other players. The ERC-20 token \texttt{SILK} can be used to buy \texttt{Tanks}.

\textbf{Sandbox}~\cite{sandbox}, like Decentraland, is a virtual world residing on Ethereum that players can build. The name was initially a 2D game's name for mobile devices and PC and then used by the NFT game when it was launched on November 2021. The NFT of Sandbox is \texttt{LANDS}, which represents the ownership of the physical spaces in the metaverse.

\section{Earn-without-payment\protect}
\label{sec:earn-wo-payment}
Today most NFT games also hold events and allow players to participate and get token rewards for free, which we call earn-without-payment. Specifically, there are three ways to earn without the investment of money: (1) participation and feedback,  (2) contributor,  and (3) virtual employee (metaverse). 
In these cases, players do not need to pay or invest money to earn, and the profits are exactly the incomes. However, the income may be the exchange for the player's time, effort, and labor. Table \ref{tbl:earn-1} lists the verified incentives (earnings without payment) in the 12 NFT games.

\textbf{Participation and feedback.} Players may get tokens when participating in in-game events, activities, and contests. For example,  Sandbox players may get rewards after completing the event requirements. Coinbase~\cite{coinbase} gives a \$ 3 in SAND to players who watched an introduction video of Sandbox and completed a quiz. Note that not all events provide rewards and the amount of the reward can not be estimated exactly, but generally is small. In addition, players who provide feedback \eg completing a survey of Decentraland~\cite{decentraland}, can earn \$5 in MANA. This is completely legit and without any risks. However, the amount of money is tiny.

\textbf{Contributor.} Decentralized autonomous organization (DAO)~\cite{dao} is the management structrue of the NFT games.  In order to be autonomous, the DAO of the NFT games typically owns a substantial amount of governance tokens. As the management structure, DAO also provides token rewards for players who contribute to the game and the community. For instance, the Decentraland~\cite{decentraland} sets up a \texttt{Community Grants} program~\cite{CommunityGrants} to reimburse the efforts of the players who contributed to the maintenance of the community or the platform and created high-quality content. There are six different tiers of rewards possible, from MANA tokens to \$240,000 USD. Sandbox also has a similar \texttt{the Ambassador Program}~\cite{TheSandboxAmbassadorProgram}, via which players can earn LAND or SAND tokens, exclusive ASSETs, and others. 
There is no doubt that contributors to NFT games require professional skills and only a few players are capable of obtaining such rewards. The majority of NFT game players may only get income through participating the events and activities instead. 
Both the opportunity of working inside the NFT games and the rewards of working as a contributor to the NFT games are limited, and there is no standard for such incomes. 

\textbf{Virtual employee (metaverse).} Only NFT games related to the metaverse recruit the players as the virtual employee (metaverse). News reported that casinos in Decentraland would like to hire players as staff~\cite{decentralandwork}. Under such a circumstance, players who get a job in the NFT game would get paid for their jobs.
For this kind of sutiation, the positions may be limited and the salary standard is opaque. In addition, the lack of related regulations and laws for such virtual positions puts the players at risk of labor abuse and fraud.

\emph{In summary, players can earn limited incomes in tokens without any money out-of-pocket in the NFT games, and the amount of income is expected to be limited and unpredictable. }

\begin{table*}
\small
\caption{Incentives in the NFT games}
\label{tbl:earn-1}
\begin{center}
\hspace*{-1.3cm}
\begin{tabular}{c|c|c|c|c|c}
NFT Game & Incentive Form & \makecell{Estimated \\ Profit (USD)} & Via & Frequency & Precondition \\\hline
BENJI \cite{benE} & Utility token & $0.0031$ & Signup & Non-repeatable & N/A \\
BENJI~\cite{benE} & Utility token & $0.0031$ & Referral & Unlimited & N/A \\
BENJI~\cite{benE} & Utility token & $0.104$ & Participation & Monthly & NFT holder\\
BLNKS~\cite{bbpE} & NFT & $2\sim5$ & Rewards & Every 60 days & Purchase\\
BLNKS~\cite{bbpE} & Utility token & $3$ & Participation & Every 60 days & N/A\\
CK~\cite{cryptokitties} & NFT & $24.8$ & Breed & Non-repeatable & NFT holder\\
EMONA~\cite{emonE} & Utility token & Up to $4.75$ & Competition & Once & NFT holder\\
EMONA~\cite{emonE} & ETH & $2308.15$ & Raffle & Once & NFT holder\\
LANDD~\cite{dlandE} & Utility token & Up to 10k & Competition & Once & NFT holder\\
LANDD~\cite{dlandE} & Utility token & Up to 5k & Competition & Once & NFT holder\\
LANDD~\cite{dlandE} & USD & $5$ & Participation &  Non-repeatable &N/A \\
LOK~\cite{lokE} & Utility token & $5\%$ & Land Ownership & N/A & NFT holder\\
LOK~\cite{lokE} & Utility token & Uncertain & Farming Resources & N/A & N/A\\
LOK~\cite{lokE} & Utility token & Uncertain & Participation & N/A & N/A\\
MCHH~\cite{mycryptoheroes} & NFT & Uncertain & Compete against others & N/A & NFT holder\\
MCHH~\cite{mycryptoheroes} & NFT & Uncertain & Join land & N/A & Purchase\\
MCHH~\cite{mycryptoheroes} & NFT & $10\%$ & Land ownership & N/A & NFT holder\\
MCHH~\cite{mycryptoheroes} & NFT & Uncertain & Quests & N/A & N/A\\
LANDS~\cite{sandbox} & Utility token & Uncertain & NFT Creation & N/A & N/A\\
LANDS~\cite{sandbox} & Utility token & Uncertain & Rent & N/A & NFT holder\\
LANDS~\cite{sandbox} & USD & $5$ & Participation & Non-repeatable & N/A\\
SOR~\cite{sorare} & Utility token & Uncertain & Competition & N/A & NFT holder\\
SOR~\cite{sorare} & NFT & $0.0238$ & Auction & Every 5 auctions & Purchase\\
SOR~\cite{sorare} & NFT & Uncertain & Referral & N/A & Referral\\
SOR~\cite{sorare} & Utility token & Uncertain & Trade & N/A & Trade\\
TANKS~\cite{spidertanks} & Utility token & Uncertain & Rent & N/A & NFT holder\\
TANKS~\cite{spidertanks} & Utility token & Uncertain & Winning games & N/A & Purchase\\
AXIE~\cite{Axie} & Utility token & Uncertain & Battles and quests & N/A & Purchase\\
AXIE~\cite{Axie} & NFT & $7.47$ & Lend & N/A & NFT holder\\
AXIE~\cite{Axie} & Utility token & Uncertain & Staking & N/A & token holder\\
\end{tabular}
\end{center}
\end{table*}

\section{Utility Token Trading, Staking, and Functioning}\hfill
\label{sec:utility}

\begin{figure}
\centering
\includegraphics[width=.9\linewidth]{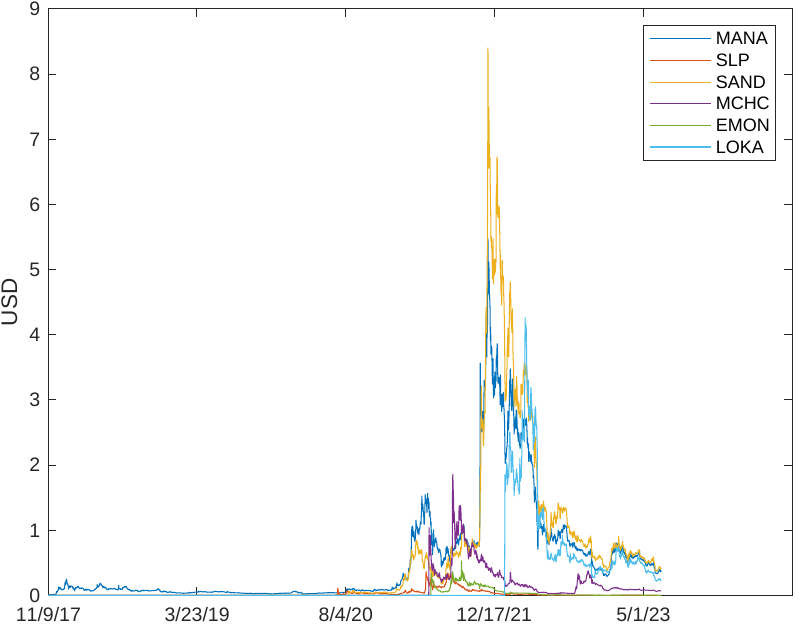}
\caption{The USD price trend of the utility tokens of the Benji Bananas (BENJI), Decentraland (MANA), Ethermon (EMON), My Crypto Heroes (MCHC), Sandbox (SAND), and Axie Infinity (SLP).}
\label{fig:ut_trend}
\end{figure}

\begin{figure}
\centering
\includegraphics[width=.9\linewidth]{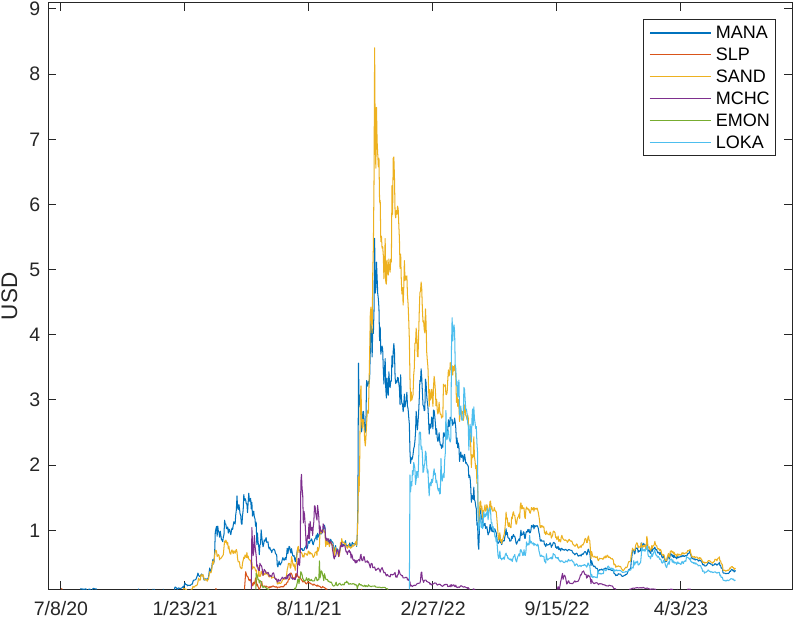}
\caption{The USD price trend of the utility tokens of the Benji Bananas (BENJI), Decentraland (MANA), Ethermon (EMON), My Crypto Heroes (MCHC), Sandbox (SAND), and Axie Infinity (SLP).}% from 7/8/20 to 6/30/23.}
\label{fig:ut_short}
\end{figure}

Besides property tokens, there are also utility tokens, including governence tokens, in most NFT games. These utility tokens can bring incomes for the owners via trading, staking, and functioning. 

\textbf{Utility Trading}
Utility tokens perform the functionality of currency inside the NFT games. Most utility tokens of NFT games are open to exchange with cryptocurrencies~\eg BitCoin~\cite{bitcoin2008}, Ether~\cite{ethereum} or currencies~\eg USD. 

We retrieve the daily exchange rate of the utility tokens of these games with USD from Yahoo! Finance~\cite{yahoo}, and the price trend is shown in Figure~\ref{fig:ut_trend} and Figure~\ref{fig:ut_short}.
It is clear that from the end of 2020 to 2022, 
the market of NFT was active when the utility tokens of different games showed a similar price trend. Starting in 2022, all games suffered from the jumped price. For example players who invest MANA before 02/10/2021 and SAND before 02/28/2021 have a chance to earn. The players who invest in SLP on 01/01/2023 can guarantee a positive return because the price flows up and down and the price on 01/01/2023 is the lowest until 03/21/2023.

\textbf{Utility Staking}
Players who are staking utility/governance tokens, especially on the crypto exchange platforms, for a certain period will get paid with interest in terms of the token she/he is staking. For example, the maximum APY found for staking the three utility tokens is 3.04\% (MANA), 14.5\% (SAND), and 2.12\% (SLP). However, Figure~\ref{fig:ut_trend} indicates % \bo{based on the figure~\ref{fig:ut_trend},} 
the price of utility tokens has dropped significantly since the end of 2021. The token interests may not cover the loss of the holders caused by the price dropping. Players can hardly earn this way. 

\textbf{Utility Functioning}
Similar to the NFT functioning, governance tokens are always bound to the voting power. In order to attract more players to participate in the management of the community, NFT games may provide rewards to players who vote. This part of earning also varies from game to game and thus we did not include in our study. 

\emph{%We can conclude that unless a player joined the game at a very early stage, the utility tokens purchased in the recent two years (since 2021) are highly likely to lead to loss to the holders~\ie players.} 
%\gyx{We can see that the price of the utility tokens had a growth period in the early summer, and the price would drop heavily starting from the fall. And the estimated profits for players who invested SLP after 07/21 and MANA or SAND after 05/2022 are negative.}}
%\csq{rephrase}
Overall, we observe that there are significant price fluctuations across these games, with a price uptake and then a downturn in a short period.  While many reasons may have contributed to this volatility, we found that players can hardly make profit via utility tokens.
}

\end{document}